\documentclass[10pt]{article}

\usepackage[T1]{fontenc}
\usepackage[utf8]{inputenc}
\usepackage{lmodern}
\usepackage{microtype}

\usepackage[margin=1in]{geometry}

\usepackage{amsmath,amssymb,amsthm}
\usepackage{aliascnt}
\usepackage{mathtools}
\providecommand{\wideparen}[1]{\overset{\frown}{#1}}
\numberwithin{equation}{section}

\usepackage{graphicx}

\usepackage{multicol}
\usepackage{enumitem}

\setlist[itemize]{topsep=2pt,itemsep=1pt,parsep=0pt,partopsep=0pt}

\usepackage{tikz-cd}

\usepackage[hidelinks]{hyperref}
\usepackage[nameinlink,noabbrev]{cleveref}
\usepackage{url}
\usepackage{doi}
\usepackage{etoolbox}
\AtBeginEnvironment{thebibliography}{%
  \raggedright
  \sloppy
}

\theoremstyle{plain}

\newaliascnt{lemma}{theorem}
\newtheorem{lemma}[lemma]{Lemma}
\aliascntresetthe{lemma}
\newaliascnt{proposition}{theorem}
\newtheorem{proposition}[proposition]{Proposition}
\aliascntresetthe{proposition}
\newaliascnt{corollary}{theorem}

\aliascntresetthe{corollary}
\theoremstyle{definition}
\newaliascnt{definition}{theorem}

\aliascntresetthe{definition}
\theoremstyle{remark}
\newaliascnt{remark}{theorem}

\aliascntresetthe{remark}

\crefname{theorem}{theorem}{theorems}
\Crefname{theorem}{Theorem}{Theorems}
\crefname{lemma}{lemma}{lemmas}
\Crefname{lemma}{Lemma}{Lemmas}
\crefname{proposition}{proposition}{propositions}
\Crefname{proposition}{Proposition}{Propositions}
\crefname{corollary}{corollary}{corollaries}
\Crefname{corollary}{Corollary}{Corollaries}
\crefname{definition}{definition}{definitions}
\Crefname{definition}{Definition}{Definitions}
\crefname{remark}{remark}{remarks}
\Crefname{remark}{Remark}{Remarks}

\begin{document}
\title{Constructive Euclidean Proofs of the Equivalence Between Keplerian Orbits and Newton's Inverse-Square Law}
\author{Changchun Shi}
\date{}

\maketitle

\begin{abstract}
  Kepler's first two laws state that a planet moves on an ellipse with the Sun at a focus and sweeps out equal areas in equal times (constant areal speed). In the \emph{Principia}, Newton showed how these laws connect to universal gravitation. Since then, the equivalence between orbital laws and force laws has remained a central topic in celestial mechanics. We present fully geometric proofs---built from explicit Euclidean straightedge-and-compass constructions---of this equivalence in both directions. The proof system combines finite-step constructions, tangent and triangle geometry, affine transport, local displacement ratios, conic invariants, and several hodograph realizations. Within this broader framework, one contribution is to use the auxiliary circle as the primary hodograph proxy in configuration space rather than the directrix-circle normalization of radius \(2a\). Our emphasis is a \emph{Principia}-style argument that avoids differential equations while remaining close to Euclidean methods.
\end{abstract}

\section{Introduction}

Kepler, building on Tycho's observations, first extracted empirical regularities: planets move on ellipses with the Sun at a focus, sweep out equal areas in equal times, and across planets satisfy the invariant period--size relation $T^2\propto a^3$ \cite{ostermannwanner2012geometry}.

In Book~I of the \emph{Principia} (1687), Newton formulates the key hinge in two parts. Proposition~I (Theorem~I) and Proposition~II (Theorem~II) establish the equivalence between central direction and area sweep proportional to time in planar motion; with the conic-focus condition added, Newton then derives inverse-square dependence \cite{newton1846mottewikisource,chandrasekhar1995principia}. In the first-edition tradition, the conic-specific historical forward chain is commonly cited as Propositions~XI--XIII, while the complementary historical reverse reconstruction chain is often cited through Propositions~XVII and~XLI. Unless noted otherwise, proposition/theorem references in this paper follow the 1726 third-edition sequence as presented in the Motte--Cajori translation line \cite{newton1846mottewikisource,newton1934mottecajori}. Chandrasekhar's common-reader commentary is widely associated with this translation tradition \cite{chandrasekhar1995principia,newton1934mottecajori}.

Historically, orbit-to-force was the direction of discovery, while force-to-orbit was its complement. In the terminology used in this paper, these are the inverse problem and forward problem, respectively. Newton's geometric arguments in Book~I treat both directions for ellipse, parabola, and hyperbola through conic-specific propositions. Under inverse-square centripetal attraction, trajectories are conic sections with the force center at a focus, and the orbit type (ellipse/parabola/hyperbola) is determined by the motion regime (equivalently, by the constants fixed by initial data). In revisions prepared around 1712 and published in the second edition (1713), Newton also included additional forward-facing comparisons that refer motion to the conic center (center of ellipse or hyperbola), complementing the focus-centered constructional route \cite{Nov95,chandrasekhar1995principia}. Modern reconstructions of Newton's argument structure are given in \cite{chandrasekhar1995principia,markowsky2011newton}.

The continental differential-calculus tradition associated with Leibniz and Bernoulli emphasized analytic generality, while Newton's treatment remained geometric and limit-based (first and last ratios). Johann Bernoulli and others criticized parts of Newton's forward-problem argument, and this methodological dispute shaped later celestial mechanics \cite{guicciardini1999reading}. A modern defense of Newton's internal deductive structure is given by Chandrasekhar (especially pp.~111--112), and a related historical assessment is given by Arnol'd \cite{chandrasekhar1995principia,arnold1990huygens}. In modern terms, fully explicit ODE existence/uniqueness theorems came later, especially with 19th-century work led by Cauchy.

Modern pedagogy usually starts from Newton's laws and a force model as given, then derives orbital laws as consequences. In the rest of this paper, we use forward/inverse in this modern sense unless explicitly marked as historical. Thus Kepler-to-$1/r^2$ is the modern inverse direction, while $1/r^2$-to-conic is the modern forward direction. Standard modern routes for the former are summarized in \cite{chandrasekhar1995principia,ostermannwanner2012geometry,markowsky2011newton}; for the latter, a common pedagogical line runs through Hamilton (1847), Maxwell (1877), and Feynman's 1964 ``lost lecture,'' and is further refined in vHH (2009) and CRS (2016) \cite{hamilton1847hodograph,maxwell1877mattermotion,goodstein1996feynman,vanhaandelheckman2009teaching,carinena2016newlook}. Our aim is to present geometric proofs in both directions, in a style close to Newton's Euclidean constructive approach.

We briefly summarize what the main cited references emphasize, to clarify how the present proof fits among them (we refer to the bibliography for full details). For historical discussion of how Book~I evolves across different versions/editions, see also \cite{Nov95,chandrasekhar1995principia}.
\begin{itemize}

  \item Nobel laureate Chandrasekhar (1995) \cite{chandrasekhar1995principia}. A detailed modern reconstruction of Newton's arguments, combining synthetic geometry with contemporary analytic language. Some derivations do not follow Newton's original route step by step, so this book is best read together with the \emph{Principia} itself. Nonetheless, it is indispensable and greatly aids reading the \emph{Principia}.

  \item Guicciardini (1999) \cite{guicciardini1999reading}. A historical study of the post-\emph{Principia} debate on Newton's mathematical method, including Bernoulli-era criticism and the geometry-vs-analysis fault line.

  \item Arnol'd (1990) \cite{arnold1990huygens}. A historical-mathematical perspective on the Newton--Hooke/Huygens--Barrow line, often cited in discussions of how Newton's geometric method should be interpreted against later analytic standards.

  \item Sir William Rowan Hamilton (1847) \cite{hamilton1847hodograph}. Hamilton introduces the hodograph and proves circular hodographs for inverse-square central attraction. This is a landmark exposition, notable for its distinctive communication style: almost no diagrams and very light symbolic machinery, yet a strikingly clear geometric argument. Hamilton is an early, prominent British-Isles figure in the 19th-century analytical reform and re-connection with continental methods, while still keeping a strongly synthetic/geometric mode of thought; this work is a prime example of that style.

  \item Maxwell (1877) \cite{maxwell1877mattermotion}. In this introductory treatment of motion and matter, Maxwell develops the pedagogical hodograph viewpoint and emphasizes its value for reconstructing orbital geometry from inverse-square dynamics. He gives an early clear construction in which a rotated hodograph circle corresponds to the ellipse's directrix circle, making the conic-to-inverse-square (modern inverse problem) derivation especially transparent for students.

  \item Goodstein \& Goodstein (1996) \cite{goodstein1996feynman}. A geometric narrative inspired by Feynman's 1964 ``lost lecture,'' emphasizing a nearly calculus-free route for the forward problem (inverse-square $\Rightarrow$ conic orbits) using hodograph synthesis.

  \item Derbes (2001) \cite{derbes2001reinventing}. A clear pedagogical account of hodographic (velocity-space) methods for the Kepler problem, including the parabolic case.

  \item Stavek (2019) \cite{stavek2019newtonsparabola}. A geometric exploration centered on Newton's parabola and related classical constructions (directrix, pedal curve, evolute, subtangent/subnormal, and a Ptolemy-circle/hodograph viewpoint).

  \item Markowsky (2011) \cite{markowsky2011newton}. A careful retelling and streamlining of Newton's arguments, with an emphasis on geometric structure and pedagogy. Markowsky also emphasizes that whether Newton regarded ``having a solution'' as sufficient for the forward problem (inverse-square $\Rightarrow$ conics) is debatable: one can justify the step via existence/uniqueness results for the associated ordinary differential equations. It also seeks a \emph{Principia}-style proof of energy conservation, from which one can derive that the trajectory is a conic.

  \item Provost \& Bracco (2008) \cite{provostbracco2008simple}. A hybrid yet insightful and concise derivation highlighting conserved quantities---essentially the Laplace--Runge--Lenz vector---that analytical approaches often use to show that the path is a conic. This paper offers an elementary view based on areal speed and related geometric relations, and also summarizes Jakob Hermann's early-18th-century proof (c.\ 1710s) and the aforementioned proof of Hamilton (1847).

  \item Woan \cite{a1dynamicalastronomy}. A standard mechanics approach in polar coordinates, deriving the conic form by solving for $r(\theta)$ using angular momentum and energy, in the tradition of analytic approaches developed by mathematicians such as Bernoulli (late 17th to early 18th century).

  \item Simha (2021) \cite{simha2021algebra}. A calculus-light, trigonometric presentation of Kepler's first law; it substantially overlaps with the Markowsky-style approach, possibly without awareness of that earlier source. It is mentioned here because it uses the law of cosines effectively.

  \item van~Haandel \& Heckman (2009) \cite{vanhaandelheckman2009teaching}. A freshman-oriented teaching note emphasizing geometric intuition. Besides giving a quick account of Feynman and Newton's methods, it also provides a geometric construction of the LRL vector that offers further insight into the forward problem.

  \item Cari\~nena--Ra\~nada--Santander (2016) \cite{carinena2016newlook}. A modern refinement of the hodograph/Feynman line (referred to in this paper as CRS16), including a cleaner force-center transformation framework.
\end{itemize}

In Section~2 we briefly repeat Newton's argument (with his classic figure) showing why the area theorem implies a centripetal (central) force. In Sections~3--5 we present geometric proofs for inverse-square dependence using the auxiliary circle, affine transport, and conic-specific refinements. Section~6 then treats the complementary forward problem through a translated, $\Delta t$-scaled hodograph construction.

\section{Central Force and Related properties}

Consider planar motion of a point mass about a fixed center, denoted by $F$ (or $S$ in Newton's polygonal construction).

Assume the areal law about the center: equal areas are swept in equal times. In local form, there is a constant $c$ such that
\begin{equation}
  \Delta s = c\,\Delta t.
\end{equation}
Equivalently, over a finite time interval $t$, one has $s=ct$.

It is convenient to write $L:=2c$ (twice the constant area rate), so
\begin{equation}
  \Delta s=\frac{L}{2}\Delta t,
  \qquad
  s=\frac{L}{2}t.
\end{equation}
In modern mechanics, $L$ is the \emph{specific angular momentum} (angular momentum per unit mass), often denoted by $h$ (see, e.g., \cite{chandrasekhar1995principia,goodstein1996feynman}). The ordinary angular momentum is then $\mathcal{L}=mL$.

We assume the motion is smooth enough that ``short time step'' expansions are valid in the usual limiting sense (errors are higher order in the small time step).
% !TEX root = ../../main.tex
\subsection{Central direction from constant areal speed}

\subsubsection{Newton's area-law centripetal direction}
\label{subsec:central-direction}

We briefly recall Newton's original geometric implication of Kepler's area law (1609) in the \emph{Principia} (Book~I, Proposition~1, 1687) \cite{newton1846mottewikisource}. We do not repeat Newton's elementary geometric proof here; an accessible modern retelling is in Section~22 (``The Area Theorem'') of Chapter~5 of \cite{chandrasekhar1995principia}. A short online exposition in the same geometric spirit is also available in \cite{nonagonnewtonkepler}, and a detailed fill-in of some steps in Newton's Proposition~1 is given in \cite{nauenberg2003keplerarea}.

\begin{figure}[t]
  \centering
  \includegraphics[width=0.5\linewidth]{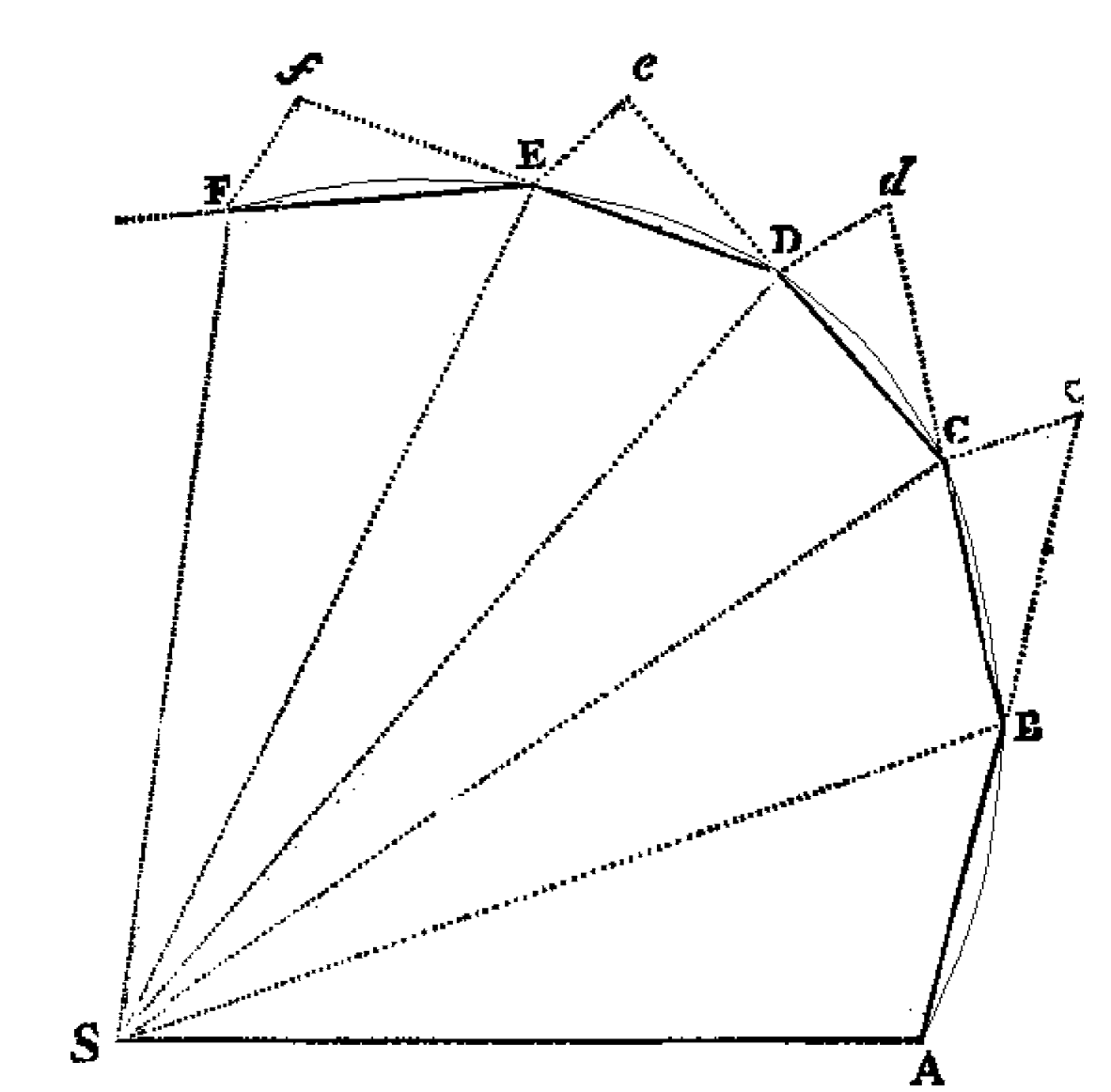}
  \caption{Newton's area theorem construction, reproduced from \cite{newton1687principia,newton1846mottewikisource}. Because it is an original 1687 diagram, it naturally looks coarse/ancient, but the geometric idea is clear.}
  \label{fig:newton-area}
\end{figure}

Approximate the orbit by a broken line $ABCDEF\cdots$ traversed in equal time steps, with instantaneous ``impulses'' at the vertices. Fix the point $S$ about which areas are swept. If there is no impulse at $B$, then after $AB$ the motion would continue straight to a point $c$ on the extension of $AB$ with $Bc=AB$. Newton's construction draws through $c$ a line parallel to $SB$ and places $C$ on this line so that $C\in SC$ (see \Cref{fig:newton-area}); this makes the successive swept triangles $\triangle SAB$ and $\triangle SBC$ have equal areas. Conversely, if equal areas are swept in equal times, then at each vertex the ``jump'' from the straight continuation (e.g., from $c$ back to $C$) must lie along $BS$ (and similarly at $C,D,\ldots$). In the smooth limit, the acceleration is therefore always directed along $AS$, i.e., the force is \emph{centripetal} with center $S$.

Thus it remains to determine the magnitude of the centripetal acceleration $a(r)$.

\subsubsection{Local distance relationships for \texorpdfstring{$BD$ and $BR^2$}{BD and BR squared}}
\label{subsec:central-direction-lemmas}

\begin{figure}[t]
  \centering
  \includegraphics[width=1\linewidth]{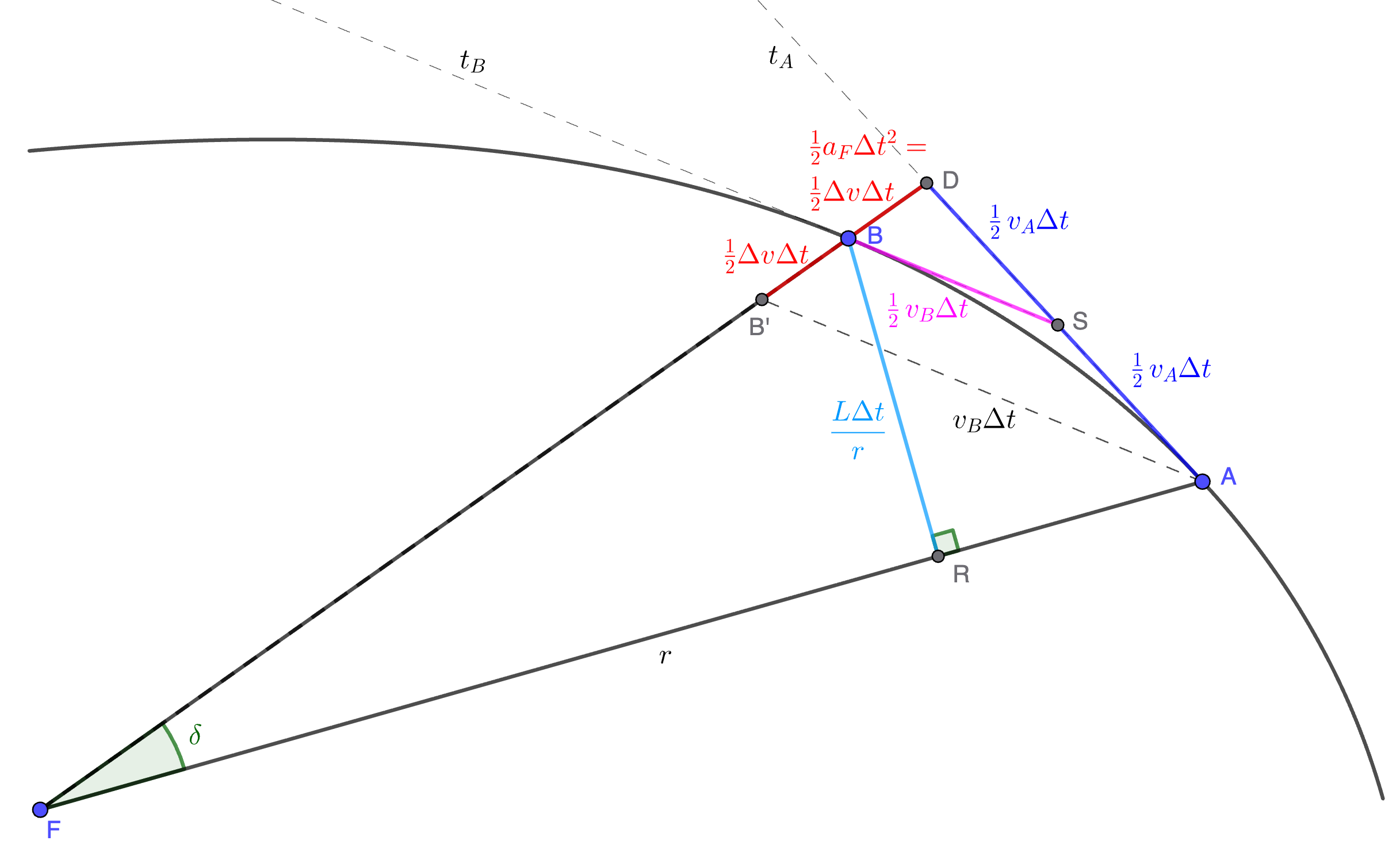}
  \caption{Local construction used in \Cref{lem:deflection-identity,lem:bd-over-br2}.}
  \label{fig:lemma-bd-br2}
\end{figure}

\begin{lemma}[Deflection identity]
\label[lemma]{lem:deflection-identity}
Over an evanescent time step $\Delta t$, the deflection produced by the mean acceleration $a_F$ satisfies
\begin{equation}
  BD=\frac12 a_F(\Delta t)^2=\frac12 \Delta v\,\Delta t,
  \qquad (\Delta v=a_F\Delta t).
\end{equation}
With the tangent construction, $S$ bisects $AD$; and in the limit $B\to A$,
\begin{equation}
  AS=SD=SB=\frac12\,\wideparen{AB}.
\end{equation}
\end{lemma}

\begin{proof}
From $\Delta v=a_F\Delta t$,
\[
  \frac12 a_F(\Delta t)^2
  =\frac12 (a_F\Delta t)\Delta t
  =\frac12 \Delta v\,\Delta t.
\]
Thus the deflection in time $\Delta t$ equals the distance generated by the mean velocity increment $\tfrac12\Delta v$ during $\Delta t$, namely the constructed segment $BD$. By construction, $S$ is the midpoint of $AD$, so $S$ bisects $AD$. As $B\to A$, the chordal midpoints coalesce in the ultimate ratio, giving
\[
  AS=SD=SB=\frac12\,\wideparen{AB}.
\]
\end{proof}

\begin{lemma}[Ratio \texorpdfstring{$BD/BR^2$}{BD over BR squared} and inverse-square form]
\label[lemma]{lem:bd-over-br2}
Let $F$ be the center of force, $r=FB$, and $BR\perp FB$. Let $L$ be twice the areal speed, so
\[
  \Delta A=\frac{L}{2}\Delta t.
\]
Then, as $\Delta t\to 0$,
\begin{equation}
  BR\sim \frac{L\Delta t}{r},
  \qquad
  \frac{BD}{BR^2}\sim \frac12\,a_F\,\frac{r^2}{L^2}.
\label{eq:bd-br2-asymptotic}
\end{equation}
Define
\begin{equation}
  \kappa:=\lim_{\Delta t\to 0}\frac{BD}{BR^2}.
\label{eq:kappa-def}
\end{equation}
Hence
\begin{equation}
  a_F=\frac{2\kappa L^2}{r^2}=\frac{\mu}{r^2},
  \qquad
  \mu:=2\kappa L^2,
  \qquad
  \vec{a}_F=-\frac{\mu}{r^2}\,\hat{r}.
\end{equation}
\end{lemma}

\begin{proof}
From Prop.~I \cite{newton1846mottewikisource}, swept areas are proportional to times; for an evanescent step,
\[
  \Delta A\sim \frac12 r\,BR.
\]
Using $\Delta A=\frac{L}{2}\Delta t$ gives
\[
  \frac{L}{2}\Delta t\sim \frac12 r\,BR
  \quad\Longrightarrow\quad
  BR\sim \frac{L\Delta t}{r}.
\]
By \Cref{lem:deflection-identity},
\[
  BD=\frac12 a_F(\Delta t)^2.
\]
Substitute $BR\sim L\Delta t/r$ to obtain
\[
  \frac{BD}{BR^2}
  \sim
  \frac{\frac12 a_F(\Delta t)^2}{(L\Delta t/r)^2}
  =
  \frac12\,a_F\,\frac{r^2}{L^2}.
\]
So if the ratio tends to $\kappa$, then $a_F=2\kappa L^2/r^2$. Writing
\[
  \mu:=2\kappa L^2,
\]
this is $a_F=\mu/r^2$, directed toward $F$, i.e.
\[
  \vec{a}_F=-\frac{\mu}{r^2}\,\hat{r}.
\]
\end{proof}

\section{Auxiliary circle of an ellipse and the affine transformation view}

Let $\ell$ be the major axis and let $O$ be the center of the ellipse. Recall that the ellipse has semi-major axis $a$ and semi-minor axis $b$. Consider the \emph{major-axis circle} $\mathcal C$: the circle centered at $O$ with radius $a$ (so its diameter equals the major axis length $2a$). This is also the standard \emph{auxiliary circle} associated to the ellipse.

Define the affine map $\Phi$ geometrically as follows. For any point $P$ in the plane, let $P_T$ be the foot of the perpendicular from $P$ to $\ell$. On the same perpendicular line $P_TP$ choose a point $P_0$ on the same side of $\ell$ as $P$ such that
\begin{equation}
  P_TP_0=\frac{b}{a}\,P_TP.
\end{equation}
Then set $\Phi(P):=P_0$.

Equivalently, $\Phi$ fixes $\ell$ pointwise and scales all lengths perpendicular to $\ell$ by the constant factor $b/a$. This map sends lines to lines, preserves parallelism, and scales all areas by $b/a$. In the standard auxiliary circle construction, $\Phi$ maps $\mathcal C$ to the ellipse.

\paragraph{Tangents via secants (Newton's 1687 viewpoint).}
Newton often treats a tangent as the \emph{ultimate position} of a secant: if $A$ is a point on the curve and $B$ is a nearby point, then the chord line $AB$ approaches the tangent line at $A$ as $B\to A$ (in Newton's 1687 language, as one takes the ``ultimate ratio''). In our setting this provides an alternative justification for the fact that $\Phi$ carries circle tangents to ellipse tangents.

Indeed, let $A_0\in\mathcal C$ correspond to $A=\Phi(A_0)$ on the ellipse, and let $B_0\in\mathcal C$ correspond to $B=\Phi(B_0)$. Since $\Phi$ is affine, it maps the secant (chord) line $A_0B_0$ to the secant line $AB$, and it preserves incidences and parallelism. As $B_0\to A_0$, the secant line $A_0B_0$ approaches the tangent to $\mathcal C$ at $A_0$, while $AB$ approaches the tangent to the ellipse at $A$. Therefore, in the same limiting (secant-to-tangent) sense, $\Phi$ maps the circle tangent at $A_0$ to the ellipse tangent at $A$.

\begin{figure}[t]
  \centering
  \includegraphics[width=\linewidth]{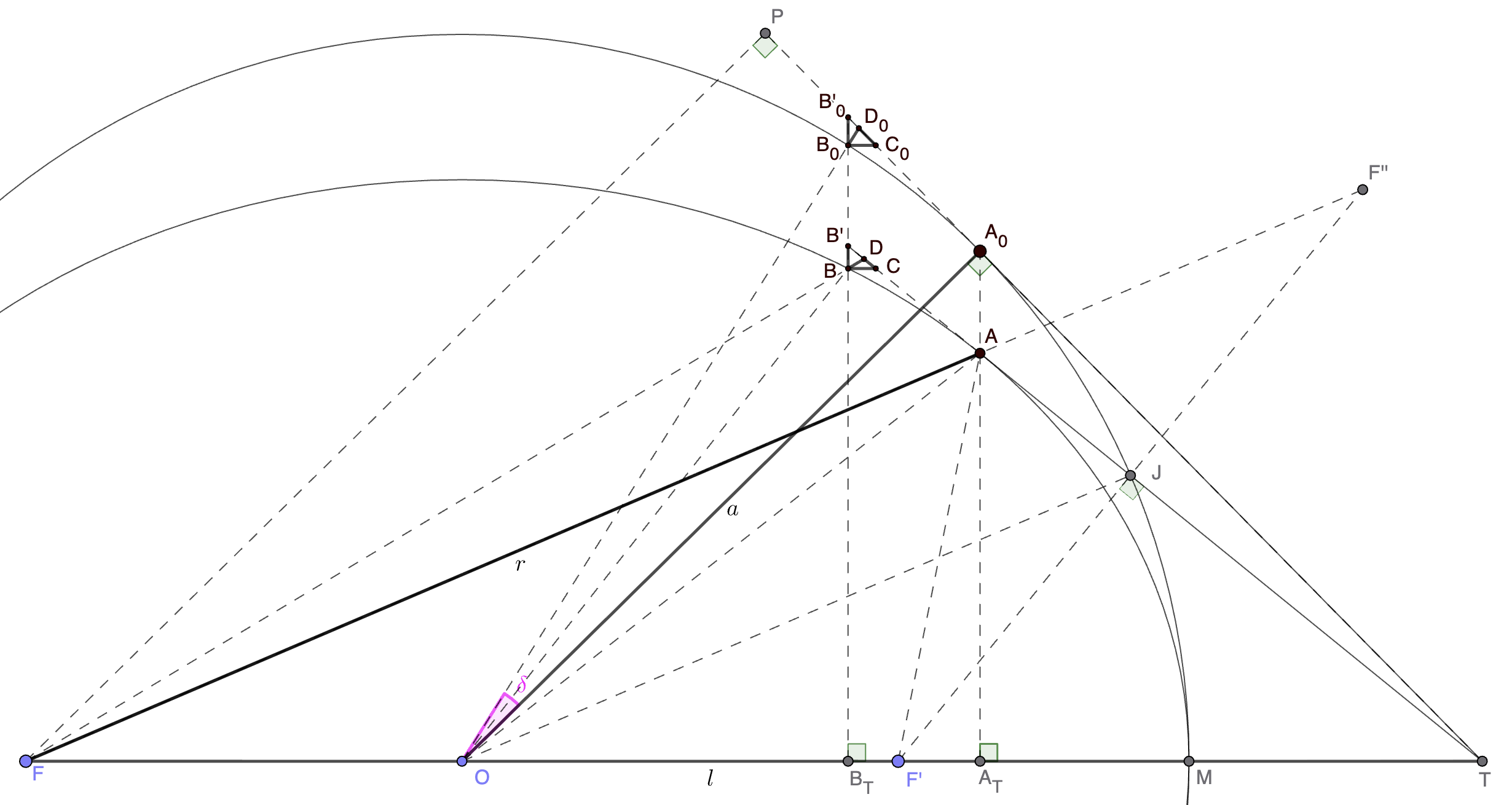}
  \caption{Affine relation of motion between the auxiliary circle and the ellipse.}
  \label{fig:affine-relation-of-motion}
\end{figure}

A few simple geometric consequences we use repeatedly in this section:
\begin{itemize}
  \item Straight lines map to straight lines under $\Phi$.
  \item Horizontal lengths (parallel to $\ell$) are unchanged by $\Phi$.
  \item Vertical lengths (perpendicular to $\ell$) scale by $b/a$.
  \item Parallel lines stay parallel under $\Phi$, and ratios of lengths along such lines are unchanged by $\Phi$.
  \item Areas scale by $b/a$.
  \item The inverse map $\Phi^{-1}$ sends $\Phi(A)$ back to $A$ and scales vertical lengths by $a/b$.
\end{itemize}
Although Newton (1687) did not formulate affine transformations explicitly, he repeatedly used equivalent geometric invariance principles. For example, Markowsky's reconstruction of Proposition~6 shows that, for an ellipse, the product of the ordinate to a conjugate diameter and the semi-length of that conjugate diameter is constant; this is an affine invariant arising from the circle-to-ellipse map \cite{markowsky2011newton}. Newton uses this same type of fact in Proposition~XI \cite{newton1846mottewikisource}. Likewise, in Book~I, Proposition~XXXI (Problem~XXIII) and its Scholium, Newton's computation of time along a given ellipse relies on area scaling between the auxiliary circle and the ellipse \cite{newton1846mottewikisource,chandrasekhar1995principia}.

\subsection{Tangent transfer and matching normal drops}

In the remainder of this section, we use \Cref{fig:affine-relation-of-motion} as the setup for understanding motion on the ellipse through its auxiliary circle. We now record a concrete straightedge-and-compass configuration that makes the affine ``transport'' of the short-time normal departure from the tangent explicit.
The purpose of this section is to obtain the geometric local-drop estimate; the kinematic conversion to acceleration then uses \Cref{lem:deflection-identity,lem:bd-over-br2} from Section~2.

\subsubsection{Circle companions and tangent transfer}

Let $\ell$ be the major axis, and let $\Phi$ map the auxiliary circle $\mathcal{C}$ to the ellipse. From $A$ and $B$ drop perpendiculars to $\ell$ and extend these same vertical lines to meet the auxiliary circle $\mathcal C$ at $A_0$ and $B_0$ (choose the intersections on the same side of $\ell$ as $A$ and $B$). By construction, $\Phi(A_0)=A$ and $\Phi(B_0)=B$.
Draw the tangent to $\mathcal C$ at $A_0$ and let it meet $\ell$ at $T$.

\begin{lemma}[Tangent transfer with fixed $T$]
  \label[lemma]{lem:tangent-transfer}
  Under $\Phi$, the circle tangent at $A_0$ maps to the ellipse tangent at $A$. Moreover, since $\Phi$ fixes $\ell$ pointwise, the intersection point $T\in\ell$ is the same for both tangents.
\end{lemma}
\begin{proof}
  $TA_0$ is the tangent to the circle at $A_0$. For points $B_0$ on the circle taken sufficiently near $A_0$, all such $B_0$ lie on the same side of the line $TA_0$.

  The affine bijection $\Phi$ sends lines to lines and half-planes to half-planes, so it preserves the relation ``lying on the same side of a line.'' Hence the line $\Phi(T)\Phi(A_0)$ meets the ellipse at $A=\Phi(A_0)$. Since $\Phi(T)=T$ and all nearby image points $B=\Phi(B_0)$ lie on one side of $TA$, the line $TA$ is tangent to the ellipse at $A$.

  Here we implicitly use the fact that a conic is smooth and therefore has a unique and well-defined tangent line at each point.\qedhere
\end{proof}

\subsubsection{The \texorpdfstring{$B,C,D$}{B,C,D} construction (ellipse and circle)}

Define $C$ as the intersection of the ellipse tangent $TA$ with the line through $B$ parallel to $\ell$, and define $D:=TA\cap FB$. Likewise, on the circle tangent $TA_0$ define $C_0$ as the intersection with the line through $B_0$ parallel to $\ell$, and define $D_0:=TA_0\cap OB_0$. By construction, $BC\parallel\ell$ and $B_0C_0\parallel\ell$.

\subsubsection{The point \texorpdfstring{$J$}{J} is on \texorpdfstring{$\mathcal{C}$}{C}}

Extend the ray $AF$ beyond $A$ and mark a point $F''$ on this ray such that $AF''=AF'$ (so $\triangle AF'F''$ is isosceles). Let $J:=AT\cap F'F''$.

\begin{lemma}[Ellipse tangent bisects the focal angle]
  The tangent line $AT$ bisects $\angle F''AF'$.
\end{lemma}

In particular (by the isosceles geometry), $AT\perp F'F''$ and $J$ is the midpoint of $F'F''$. But we also know $O$ bisects $FF'$, so $OJ\parallel FA$ and $OJ=FF''/2=a$, which proves that $J$ lies on the circle $\mathcal{C}$.

\subsubsection{Matching the normal drops}

The triangle similarity relations in the tangent line through $T$ give
\begin{equation}
  B_0D_0 = B_0C_0\,\frac{OD_0}{OT},
  \qquad
  BD = BC\,\frac{FD}{FT}.
\end{equation}
Dividing and using that $\Phi(B_0)=B$ and $\Phi$ maps the line through $B_0$ parallel to $\ell$ to the line through $B$ parallel to $\ell$ (since $\Phi$ fixes $\ell$ and preserves parallelism), while also mapping the tangent line $TA_0$ to $TA$, we have $\Phi(C_0)=C$. Hence the horizontal segment $B_0C_0$ is carried to $BC$, and since $\Phi$ preserves horizontal lengths, $B_0C_0=BC$.
\begin{equation}
  \frac{B_0D_0}{BD}=\frac{OD_0}{OT}\cdot\frac{FT}{FD}.
\end{equation}
Letting $B\to A$ (so $D_0\to A_0$ and $D\to A$), we obtain
\begin{equation}
  \lim_{B\to A}\frac{B_0D_0}{BD}=\frac{OA_0}{OT}\cdot\frac{FT}{FA}.
\label{eq:matching-drops-limit}
\end{equation}

\begin{lemma}[Matching drops]
  \label[lemma]{lem:matching-drops}
  In the infinitesimal-step limit, the circle and ellipse normal drops from the tangent agree:
  \begin{equation}
    \lim_{B\to A}\frac{B_0D_0}{BD}=1.
    \label{eq:matching-drops-claim}
  \end{equation}
\end{lemma}
\begin{proof}
  Now $F,O,T\in\ell$, and by construction $J\in AT$ and $OJ\parallel FA$. Hence triangles $\triangle TFA$ and $\triangle TOJ$ are similar. Therefore
  \[
    \frac{FT}{FA}=\frac{OT}{OJ}.
  \]
  Since $A_0,J\in\mathcal C$ we have $OA_0=OJ=a$, and consequently the right-hand side of \eqref{eq:matching-drops-limit} becomes
  \[
    \frac{OA_0}{OT}\cdot\frac{FT}{FA}=\frac{a}{OT}\cdot\frac{OT}{a}=1.
  \]
  This proves \eqref{eq:matching-drops-claim}.
\end{proof}

\subsubsection{Transporting swept area from the ellipse to the circle}
\label{sec:transport-and-sagitta}

Let $A$ be the current point on the ellipse and $B$ the position after a small time increment $\Delta t$. Let the swept area about the focus be
\begin{equation}
  S = \operatorname{area}(\triangle F A B)=\frac{L}{2}\,\Delta t.
\end{equation}

We compare this to a corresponding small triangle on the circle. Let
\begin{equation}
  S_0 := \operatorname{area}(\triangle O A_0 B_0),
\end{equation}
Two geometric scalings relate $S$ and $S_0$:
\begin{enumerate}
  \item Height scaling from $O$ to $F$. In the construction used in our discussion, the line through $F$ parallel to a fixed circle radius implies that the perpendicular height from $F$ to the chord $AB$ is $(r/a)$ times the perpendicular height from $O$ to the same chord direction. Hence
    \begin{equation}
      \operatorname{area}(\triangle F A B)=\frac{r}{a}\,\operatorname{area}(\triangle O A B).
    \end{equation}

  \item Affine area scaling. Since $\Phi$ scales area by $b/a$,
    \begin{equation}
      \operatorname{area}(\triangle O A B)=\frac{b}{a}\,\operatorname{area}(\triangle O A_0 B_0)=\frac{b}{a}\,S_0.
    \end{equation}
\end{enumerate}

Combining,
\begin{equation}
  S=\operatorname{area}(\triangle F A B)=\frac{r}{a}\cdot\frac{b}{a}\,S_0=\frac{rb}{a^2}\,S_0,
\end{equation}
so
\begin{align}
  S_0 &= \frac{a^2}{rb}\,S \label{eq:s0-from-area}
  = \frac{a^2}{rb}\cdot \frac{L}{2}\,\Delta t
  = \frac{L a^2}{2 b r}\,\Delta t.
\end{align}

\subsubsection{Circle tangent and sagitta}
\label{subsec:sagitta}

As $B_0\to A_0$, \eqref{eq:s0-from-area} gives
\begin{equation}\label{eq:a0d0}
  A_0D_0 = \frac{2S_0}{a}=
  \frac{2}{a}\cdot\frac{L a^2}{2 b r}\,\Delta t
  =\frac{L a}{b r}\,\Delta t.
\end{equation}

Now use elementary circle geometry relating the \emph{normal drop} from the circle to its tangent (the sagitta). For a circle of radius $a$, the perpendicular distance from the circle point reached to the tangent line satisfies (to leading order)
\begin{equation}
  B_0D_0 \approx \frac{(A_0D_0)^2}{2a}.
\end{equation}
Substituting \eqref{eq:a0d0} gives
\begin{equation}
  B_0D_0 \approx \frac{1}{2a}\left(\frac{L a}{b r}\Delta t\right)^2
  = \frac{L^2 a}{2 b^2 r^2}\,(\Delta t)^2.
\end{equation}

\subsubsection{Inverse problem Proof 1 (ellipse via affine transported drops)}
\label{subsec:proof1-ellipse-affine}

\begin{proposition}[Inverse problem Proof 1: ellipse via affine transported drops]
\label[proposition]{prop:proof1-ellipse-inverse-square}
Under the area-law setup of Section~2 and the ellipse geometry developed in this section, the acceleration is focus-directed and has inverse-square magnitude:
\begin{equation}
  \mathbf a_F(r) = -\mu\,\frac{\mathbf r}{r^3},
  \qquad
  \text{i.e. }\mathbf a_F(r)=-\frac{\mu}{r^2}\,\hat r,
  \qquad
  \mu=\frac{L^2 a}{b^2}=\frac{L^2}{p},
\end{equation}
where $p=b^2/a$ is the semi-latus rectum and $\mathbf r$ is the position vector from $F$.
\end{proposition}
\begin{proof}
From \Cref{lem:matching-drops} and the sagitta estimate above,
\[
  BD\sim B_0D_0\sim \frac{L^2 a}{2 b^2 r^2}\,(\Delta t)^2.
\]
Also, by the area-law relation (equivalently, the first asymptotic in \Cref{lem:bd-over-br2}),
\[
  BR\sim \frac{L}{r}\,\Delta t.
\]
Therefore
\[
  \frac{BD}{BR^2}\to \frac{a}{2b^2}=\frac{1}{2p},
  \qquad \left(p=\frac{b^2}{a}\right).
\]
Applying \Cref{lem:bd-over-br2} with $\kappa=\frac{1}{2p}$ gives
\[
  a_F=\frac{2\kappa L^2}{r^2}
  =\frac{L^2}{p}\cdot\frac{1}{r^2}
  =\frac{L^2 a}{b^2}\cdot\frac{1}{r^2}.
\]
Combining this magnitude with Section~2's central-direction conclusion yields
\[
  \mathbf a_F(r) = -\mu\,\frac{\mathbf r}{r^3},
  \qquad
  \text{i.e. }\mathbf a_F(r)=-\frac{\mu}{r^2}\,\hat r,
  \qquad
  \mu=2\kappa L^2=\frac{L^2a}{b^2}=\frac{L^2}{p}.
\]
For the areal constant $k$ with $L=2k$, this is equivalently $\mu=4k^2/p$.
\end{proof}

\subsubsection{Section Summary}
This section presents an alternative geometric route to the inverse problem in the elliptic case. The key mechanism is affine transport between the auxiliary circle and the ellipse, which turns local circle-drop estimates into the corresponding ellipse estimates and then, through Section~2's lemmas, into the inverse-square force law.

It is natural to ask how this approach extends to other conic sections. Several attempts are possible, but some are not valid without additional structure. Our current view is that the affine transformation between ellipse and circle is the essential ingredient of this proof, and that an equally direct mapping is not immediately available for parabola and hyperbola.

\section{Alternative auxiliary-circle proofs of the inverse problem}
\label{sec:alternative-proof}

In this section we give more derivations of the same $1/r^2$ law, still using the auxiliary circle $C$ (center $O$, radius $a$), but emphasizing a different set of local ellipse facts so that the argument has as few moving parts as possible. Similar to Newton's methods, these approaches apply to other forms of conic sections as well, and we will discuss this further in \Cref{sec:other-conics}.

\begin{figure}[t]
  \centering
  \includegraphics[width=\linewidth]{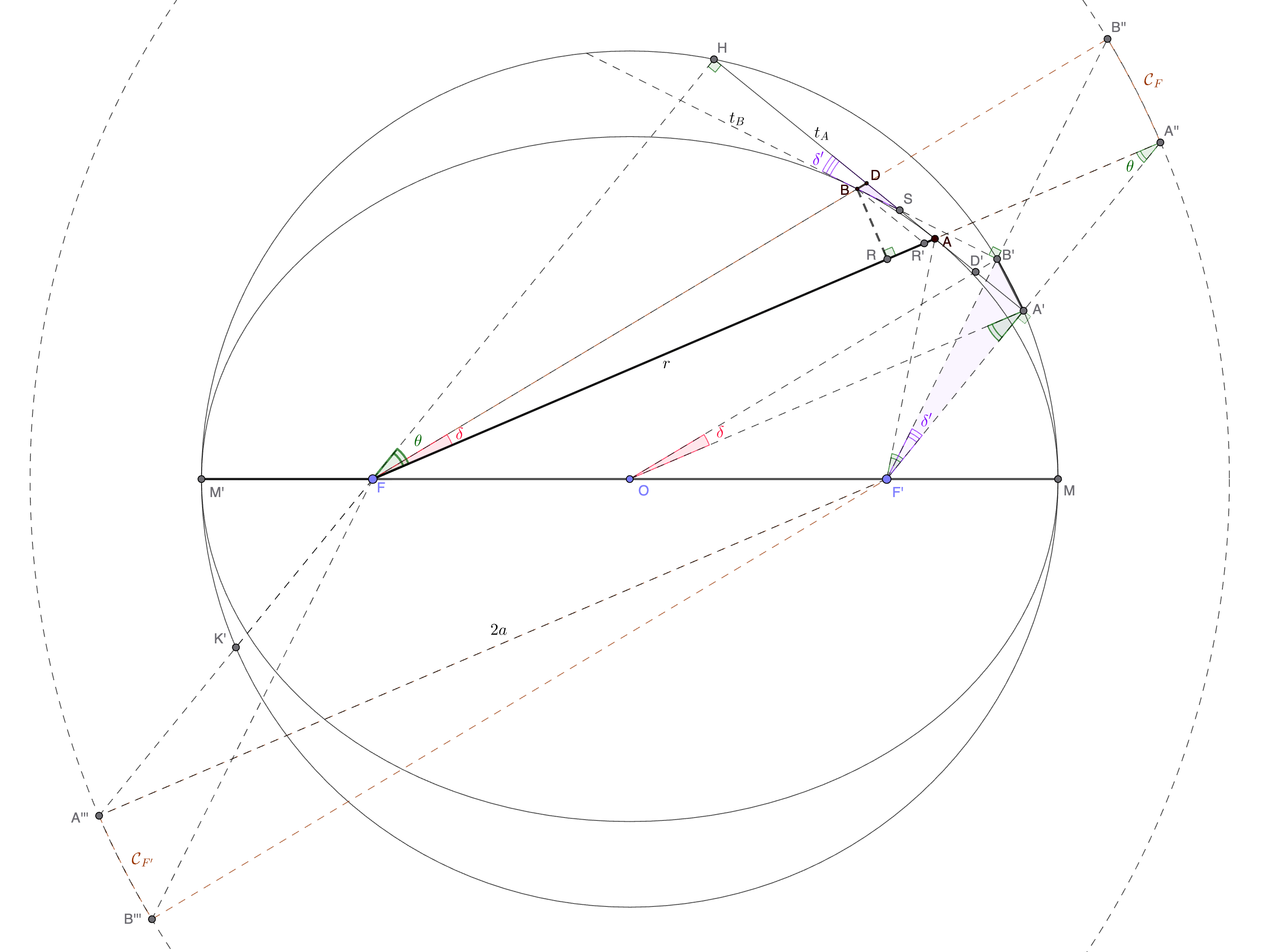}
  \caption{Alternative auxiliary-circle construction used in Section~\ref{sec:alternative-proof}.The angles $\delta$ and $\delta'$ are approaching 0 as $B\to A$. The directrix circles $\mathcal{C}_F$ and $\mathcal{C}_{F'}$ are shown for reference; they are not used in the proof but can replace the auxiliary circle as alternatives. All three circles are also alternative rotated hodographs as metnioned in \Cref{sec:forward-problem-hodograph}.}
  \label{fig:proof2-ellipse}
\end{figure}

\subsection{Geometric setup}

Let the ellipse have foci $F,F'$. Fix a point $A$ on the ellipse and write $r=FA$.
Let the tangent at $A$ meet the auxiliary circle again at $A'$. For a nearby point $B$ on the ellipse, let its tangent meet the same auxiliary circle at $B'$. We let $B\to A$.
Denote these tangent lines by $t_A$ and $t_B$, respectively; the instantaneous velocity vectors $\mathbf v_A$ and $\mathbf v_B$ are directed along $t_A$ and $t_B$.

Let $D$ be the intersection point on the tangent construction (as in \Cref{fig:proof2-ellipse}), and let $R$ be the foot used to form the small segment $BR$ (the small transverse piece in the right triangle at $R$).

\begin{lemma}[Product identity]
  \label[lemma]{lem:product-identity}
  With the notation of \Cref{fig:proof2-ellipse},
  \begin{equation}
    FH\cdot F'\!A'=(a-c)(a+c)=b^2,
  \end{equation}
  where $c=OF$ is the focal distance (so $a^2=b^2+c^2$).
\end{lemma}

\begin{proof}
  Since $FK' \parallel F'A'$ in the construction, triangles with corresponding sides along these parallels are similar, and the segment on the parallel through $F$ has the same length as the corresponding segment on the tangent through $F'$. In particular,
  \begin{equation}
    F'\!A' = FK'.
  \end{equation}
  Therefore
  \begin{equation}
    FH\cdot F'\!A' = FH\cdot FK' = (a-c)(a+c)=a^2-c^2=b^2,
  \end{equation}
  which is the standard ellipse identity.
\end{proof}

\subsection{Iverse problem Proof 2 of the ellipse case}

\begin{proof}
As $B\to A$, the little segments cut off by nearby tangents/secants agree to first order; thus the similar-triangle relations used below are valid in the Newtonian (1687) ``ultimate ratio'' sense (any error is higher order and vanishes in the limit).

Similarity in \Cref{fig:proof2-ellipse} gives
\begin{equation}
  \triangle SDB \sim \triangle F'\!A'B'
  \qquad\Longrightarrow\qquad
  BD=SD\cdot\frac{A'B'}{A'F'}.
\end{equation}
By \Cref{lem:deflection-identity} in Section~\ref{subsec:central-direction-lemmas}, $SD=\frac12\,AD$, hence
\begin{equation}
  BD=\frac12\,AD\cdot\frac{A'B'}{A'F'}.
\end{equation}

From focal-geometry and auxiliary-circle similarities (equivalently, affine scaling ellipse $\leftrightarrow$ circle),
\begin{equation}
  \begin{aligned}
    AD&=\frac{AF}{FH}\,BR=\frac{r}{FH}\,BR,\\
    A'B'&=\frac{OA'}{F'\!A}\,BR=\frac{a}{r}\,BR.
  \end{aligned}
\end{equation}
\begin{equation}
  BD=\frac12\left(\frac{r}{FH}BR\right)\left(\frac{a}{r}BR\right)\frac{1}{A'F'}
  =\frac12\,a\,\frac{BR^2}{FH\cdot A'F'}.
\end{equation}
By \Cref{lem:product-identity}, $FH\cdot A'F'=b^2$, hence
\begin{equation}
  BD=\frac12\,a\,\frac{BR^2}{b^2}.
\end{equation}
Using $p=b^2/a$,
\begin{equation}
  BD=\frac{1}{2p}\,BR^2,
  \label{eq:proof2-bd-br2}
\end{equation}
so
\[
  \frac{BD}{BR^2}\to \frac{1}{2p}.
\]
Applying \Cref{lem:bd-over-br2} with $\kappa=\frac{1}{2p}$ gives immediately
\[
  a_F=\frac{2\kappa L^2}{r^2}=\frac{L^2}{p}\frac{1}{r^2},
  \qquad
  \mathbf a_F(r)=-\mu\,\frac{\mathbf r}{r^3},
  \qquad
  \text{i.e. }\mathbf a_F(r)=-\frac{\mu}{r^2}\,\hat r,
  \qquad
  \mu=2\kappa L^2=\frac{L^2}{p}.
\]
\end{proof}

\subsection{Inverse problem Proof 3 (hodograph-style computation)}
\label{subsec:proof3-hodograph}

From \Cref{lem:product-identity},
\begin{equation}
  F'\!A'=\frac{b^2}{FH}.
\end{equation}
Let $v=\lvert \mathbf v\rvert$ be the speed at $A$. Since the velocity direction is tangent to the orbit and $FH$ is perpendicular to that tangent (as in the construction), the areal-rate constant gives
\begin{equation}
  FH\cdot v = L.
\end{equation}
Combining,
\begin{equation}
  F'\!A'=\frac{b^2}{FH}=\frac{v\,b^2}{L}.
\end{equation}
For two nearby points $A,B$ we therefore have, in the Newtonian small-step sense,
\begin{equation}
  A'B'=\Delta(F'\!A')=\frac{b^2}{L}\,\Delta v
  =\frac{b^2}{L}\,a_F\,\Delta t,
\end{equation}
where $a_F$ is the magnitude of the (focus-directed) acceleration and we used \Cref{lem:deflection-identity} ($\Delta v=a_F\Delta t$).
On the other hand, from the auxiliary-circle similarity used above,
\begin{equation}
  A'B'=\frac{OA'}{F'\!A}\,BR=\frac{a}{r}\,BR.
\end{equation}
Eliminating $A'B'$ gives
\begin{equation}
  a_F=\frac{L}{b^2}\,\frac{A'B'}{\Delta t}.
\label{eq:proof3-af-from-ab}
\end{equation}
Next, using again the auxiliary-circle similarity $A'B'=\frac{a}{r}BR$ and \Cref{lem:bd-over-br2} ($BR\sim L\Delta t/r$), we have
\begin{equation}
  A'B'=\frac{a}{r}\,\frac{L\,\Delta t}{r}
  =\frac{a}{r^2}\,L\,\Delta t.
\label{eq:proof3-ab-from-br}
\end{equation}
Combining \eqref{eq:proof3-af-from-ab} and \eqref{eq:proof3-ab-from-br} yields
\begin{equation}
  a_F=\frac{L^2}{p}\,\frac{1}{r^2},
  \qquad
  \mathbf a_F(r)=-\mu\,\frac{\mathbf r}{r^3},
  \qquad
  \text{i.e. }\mathbf a_F(r)=-\frac{\mu}{r^2}\,\hat r,
  \qquad
  \mu=\frac{L^2}{p}.
\end{equation}
exactly the same inverse-square formula as in \Cref{prop:proof1-ellipse-inverse-square}.
\hfill$\square$

\subsection{Discussion: directrix circle and a hodograph viewpoint}
\label{subsec:sec6-discussion}

\paragraph{Directrix circle variant.}
One can re-run essentially the same similarity argument by replacing the auxiliary circle $\mathcal C$ (center $O$, radius $a$) with the \emph{directrix circle} centered at $F$ of radius $2a$ (cf.\ the construction used in Section~3).
Extending $FA$ to meet this directrix circle at $A''$, one has $A''A'F'$ colinear; similarly, extending $FB$ meets the directrix circle at $B''$ with $B''B'F'$ colinear.
Moreover, $A'$ and $B'$ are midpoints of the segments $A''F'$ and $B''F'$, respectively, so
\begin{equation}
  A''B'' = 2\,A'B'.
\end{equation}
Consequently,
\begin{equation}
  \triangle A'B'F' \sim \triangle A''B''F',
\end{equation}
and the same chain of similar-triangle identities yields the same ultimate-ratio estimate for the normal drop $BD$, hence the same inverse-square dependence.

\paragraph{Connection to Maxwell's hodograph.}
The product identity in \Cref{lem:product-identity} (already implicit in classical ellipse geometry) also appears in Maxwell's 1877 discussion of the \emph{hodograph} of planetary motion: for an inverse-square central force the velocity vector $\mathbf v(t)$ traces a circle in velocity space (Hamilton's 1847 hodograph), and Maxwell notes that this hodograph circle is similar to the directrix circle, with its ``speed origin'' at $F'$, after a \(-90^\circ\) rotation (clockwise, i.e.\ opposite to the orbital sense) and a scaling that matches the directrix-circle size \cite{maxwell1877mattermotion}.

From this point of view, our construction may be read as relating the geometric drop $BD$ directly to a rotated hodograph circle $C$.
This suggests that the previous Proof 3 is essentially a repetition of Maxwell's approach to the inverse problem, but with the auxiliary circle $\mathcal{C}$ playing the role of the hodograph circle instead of the directrix circle. All three circles are alternative rotated hodographs for the same inverse-square dynamics, and any of them can be used in a similar way to relate the geometric drop $BD$ to the velocity change $\Delta v$. Moreover, the same auxiliary-cicle can also have two different interpretations treating either F or F' as the velocity origin, this is by noting that K' as the image of A' under the symmetry about center O. As argued in \cite{carinena2016newlook}, argued the two directrix circles can both be viewed as the transformed hodograph with either force F as the origin or F' as the origin; and circle $\mathcal{C}_F'$F being the rotated hodograph around focus F has additional advantage with the common polar origin. In the auxiliary-circle construction, either directrix circle is simply a scaled version of the auxiliary circle with the scale center being either F or F'.

% !TEX root = ../../main.tex
\section{Other conic sections}
\label{sec:other-conics}

So far we have focused on the elliptic (bound) case. In the same inverse-square central-force setting, the remaining conic trajectories are the parabola (threshold case) and the hyperbola (unbound case), together with one degenerate limit discussed first.

\subsection{Rectilinear motion as a degenerate conic limit}
\label{subsec:strange-line}

This extreme limit corresponds to rectilinear motion. It can be viewed as a collapsed ellipse with $a=c$ and $b=0$, i.e., a degenerate conic with zero areal speed about the force center. That case is outside the present scope: here we use nondegenerate conic geometry together with nonzero constant areal speed to solve the inverse problem. For a straight-line orbit, many different centripetal laws can produce the same geometric path, and the real task is to determine the time parametrization of motion along the line. We therefore omit this case, though Newton discusses it in detail in the \emph{Principia} using geometric parametrizations.

\subsection{Hyperbola}
\label{subsec:hyperbola}

For an inverse-square central force, hyperbolic trajectories correspond to unbound motion: the point mass comes in from infinity and escapes back to infinity, with the force center located at a focus of the hyperbola.

See \Cref{fig:hyperbola} for the auxiliary-circle construction used to carry out the same inverse-square argument in the hyperbola case.
\begin{figure}[!t]
  \centering
  \includegraphics[width=0.9\linewidth]{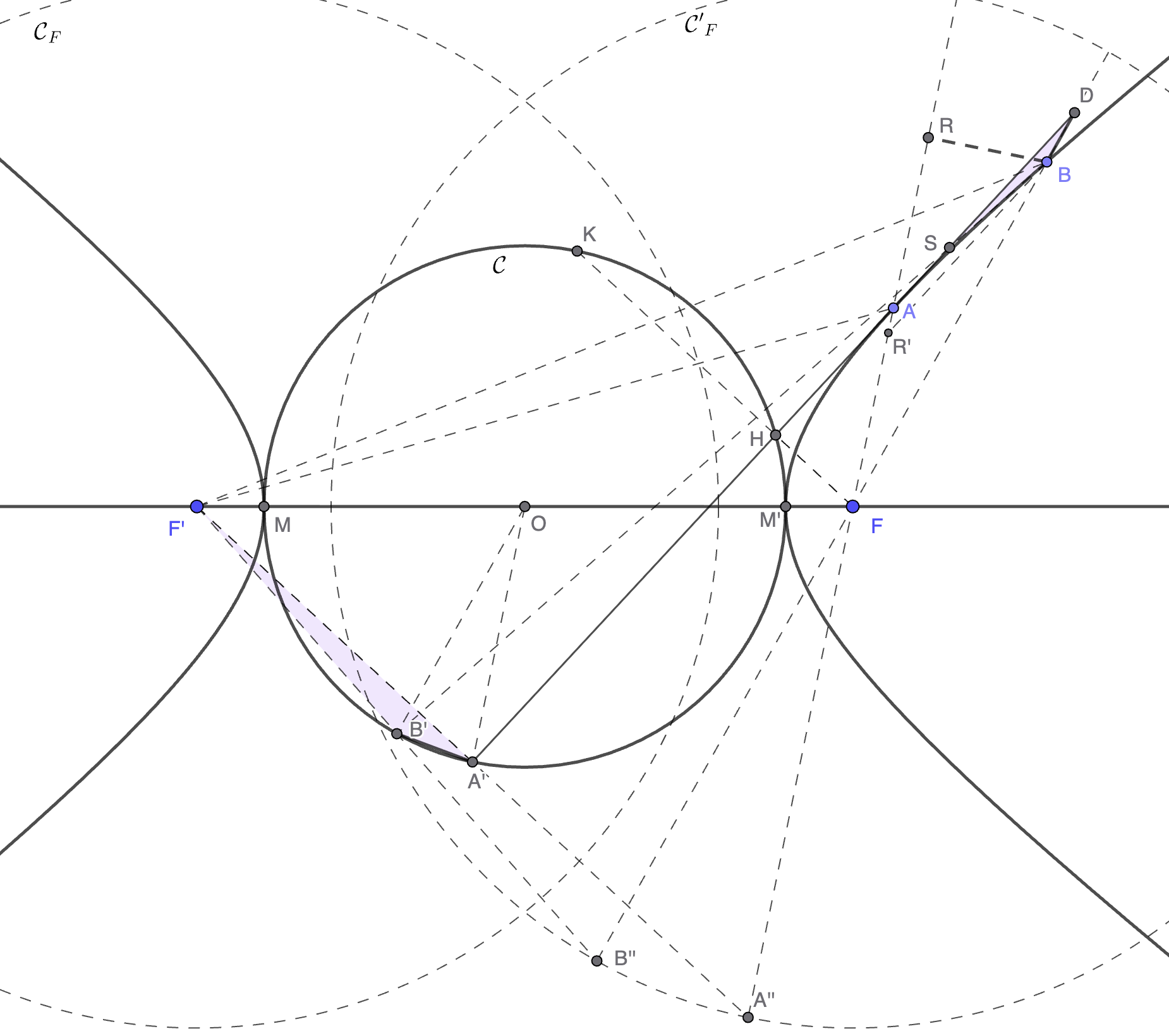}
  \caption{Auxiliary-circle construction for the hyperbola case.}
  \label{fig:hyperbola}
\end{figure}
In this hyperbola setup, one may again view the auxiliary circle $\mathcal C$ as providing a convenient geometric proxy for the (circular) hodograph, with $F'$ taken as the velocity origin.
In this hyperbola convention, this proxy picture must be rotated by \(+90^\circ\) (along the planet's direction of motion about the force center) to match the velocity direction in the hodograph construction. Compared with the ellipse \((-90^\circ)\) convention in Section~\ref{subsec:sec6-discussion}, this sign difference can be unified by convention: keep one rotation direction and absorb the sign into the scaling factor (ellipse uses a negative scale), as discussed in Section~\ref{subsec:proof2f-discussion}.

Likewise, the directrix circle $\mathcal{C}_{F'}$ may be regarded as a $2\times$ scaling of the auxiliary circle $\mathcal C$ with scaling center $F'$; consequently it can also be interpreted as a scaled version of the hodograph, obtained here after the same \(+90^\circ\) rotation.

The same key product identity (\Cref{lem:product-identity}) also holds in this hyperbola setting, though processed differently:
\begin{equation}
  F'\!A'\cdot FH = FH\cdot FK = F\!M'\cdot FM = (c-a)(c+a)=c^2-a^2=b^2,
\end{equation}
using the same tangent/intercept relations (here $c=OF$ and $a^2=b^2+c^2$).

Similarly, in the analogue of inverse problem Proof~2 the computation is again driven by the same pair of similar triangles $\triangle F'\!A'B'$ and $\triangle SDB$.

There is also a ``swapped-focus'' situation leading to the same right-branch hyperbola: instead of an attractive (centripetal) force toward $F$, one may consider a centrifugal force directed away from $F'$ acting on the particle as it moves from $A$ to $B$.
In that case to carry out the proof, the roles of $F$ and $F'$ are interchanged, and the point $R$ is dropped onto the extension of the new line $FA$.
The points $H$ and $A'$ also exchange roles, $B'$ moves to the opposite side of the auxiliary circle $\mathcal C$, and $D$ becomes the intersection of the tangent line at $A$ with the new line $FB$ (namely, $F'B$ as in \Cref{fig:hyperbola}). Readers should be able to construct the modified plot accordingly without difficulty, and the proof then follows in exactly the same way.

So once \Cref{fig:hyperbola} is in place, the remaining steps corresponding to our Proofs~2 and~3 for the hyperbola follow with only minor, mostly notational, changes. In particular, once the hyperbola geometry yields a constant limit for $BD/BR^2$, \Cref{lem:bd-over-br2} gives the inverse-square form directly.

\subsection{Parabola}
\label{subsec:parabola}

\subsubsection{Parabola as a limiting conic and hodograph setup}
\label{subsubsec:parabola-general}

The parabolic trajectory is the boundary between bound (elliptic) and unbound (hyperbolic) motion.
It is also distinguished geometrically as the conic section obtained when the cutting plane is parallel to a generating edge of the cone.

One can also view a parabola as a limiting ellipse in which $a,c\to\infty$ while $(a\!-\!c)$ stays finite.
Indeed, since $p=b^2/a=(a+c)(a-c)/a$, in this limit one has $p\to 2(a-c)$.
It is therefore natural to parametrize the parabola by setting $a-c$ as a single parameter $a'$, so that $p=2a'$, where $a'$ is the focal length (the vertex-to-focus distance). From this perspective, the auxiliary circle degenerates to the tangent line at the vertex of the parabola: at the vertex, this infinite-radius circle is a line perpendicular to the symmetry axis $AF$ (with $F$ the remaining focus).

This is consistent with the polar form: the eccentricity satisfies $e=c/a=1$ for a parabola. One can similarly think of the parabola as an approximation/limit of hyperbolas as $a,c\to \infty$ while keeping $(c\!-a)$ constant.

As the eccentricity of an ellipse tends to 1, the second focus recedes to $\infty$ along the axis; in the limit it becomes a remote focus. Hence the “directrix circle” construction from the ellipse case no longer survives as a literal circle centered at the observed focus.
Yet the effect of that construction does survive: a circle of infinite radius, centered at the remote focus, degenerates into a straight line perpendicular to the axis—namely the parabola’s directrix. In this limiting viewpoint the directrix may be regarded as the remnant of the former directrix circle, positioned so that it meets the axis
$AF$ at a point symmetrically placed with respect to the vertex
$A$ (in the same sense as in the elliptical construction). Thus, while an ellipse and a hyperbola each have two directrix circles and two directrix lines, a parabola has only one observable directrix ``circle'': the directrix line. The other lies at infinity and is not meaningful in this construction. The directrix line intersects the major axis at $2a-(a+c)=a-c=a'$ from the vertex, so it lies opposite the focus at the same distance from the vertex.

However, we will also show that a parabola still admits several natural associated circles that encode useful geometry. In particular, although the earlier auxiliary circle and directrix circle have both degenerated into a line, one may still construct rotated hodograph circles with various scalings and velocity origins, similar to those used for the ellipse and hyperbola. These remain genuinely circular and retain the kinematic and geometric information needed for our proofs. Since a parabola has only one degree of freedom $a'$, there are more relationships among angles and distances than in the generic conic sections. The parabolic case is generally considered a simpler special case of the conic sections, but in some ways it is more subtle and requires special care, as we will see in the proofs below.

In this section we show that the analogues of inverse problem Proofs~2 and~3 from the ellipse case also hold for the parabola, with some special care.
We recommend reading Newton-style arguments (in the 1687 tradition), such as the historical collection of various properties of the parabola in \href{https://doi.org/10.5539/apr.v11n2p30}{\cite{stavek2019newtonsparabola}}.
A clear explanation of the hodograph for the parabolic case is given by Derbes (2001) \cite{derbes2001reinventing}. Derbes uses the circle centered at the focus $F$ with radius $2a'$ and shows that the Maxwell (1877)/Hamilton (1847) hodograph proof for the inverse problem extends naturally to the parabola.

\subsubsection{Parabola variant of inverse problem Proof 2: a constant \texorpdfstring{$BR^2/BD$}{BR squared over BD}}
\label{subsubsec:parabola-proof2}

\begin{figure}[t]
  \centering
  \includegraphics[width=\linewidth]{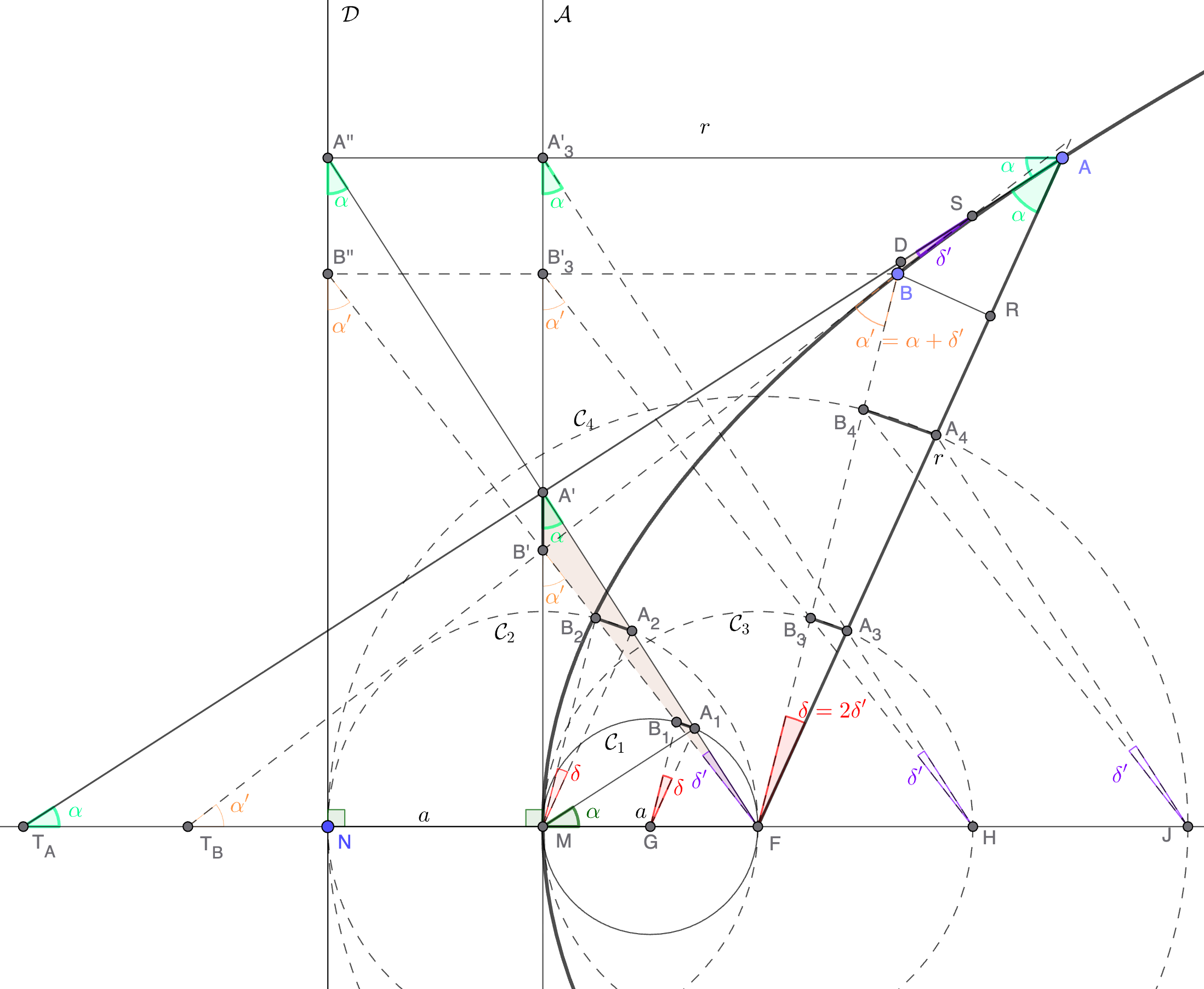}
  \caption{Parabola construction used in Section~\ref{subsec:parabola} (parabola analogue of inverse problem Proof~2).}
  \label{fig:parabola-proof2}
\end{figure}

This is the parabola analogue of inverse problem Proof~2 in the auxiliary-circle style. For convenience in the rest of the discussion we will simply use $a$ instead of $a'$ to denote the parabola parameter, so that the semi-latus rectum is $p=2a$.

\begin{proposition}[Parabola variant of inverse problem Proof~2]
\label[proposition]{prop:parabola-br2-over-bd}
In the parabolic construction (notation matching the ellipse case), as $B\to A$ one has
\begin{equation}
  \lim_{B\to A}\frac{BR^2}{BD}=4a,
\end{equation}
where $a$ is the parabola parameter (so the semi-latus rectum is $p=2a$).
\end{proposition}

\begin{proof}
As $B\to A$ (small angles $\delta,\delta'\to 0$), similar triangles in the figure give
\[
  \frac{A'B'}{BD}=\frac{B'F}{DS}\sim \frac{2\,B'F}{DA}.
  \qquad\text{(1)}
\]
Resolve the ratio with the common angle $\alpha$:
\[
  \frac{B'F}{DA}
  =
  \frac{B'F\sin\alpha}{DA\sin\alpha}
  \sim \frac{a}{BR}.
  \qquad\text{(2)}
\]
since $DA\sin\alpha\sim BR$ and $B'F\sin\alpha=a$. Substituting (2) into (1),
\[
  \frac{A'B'}{BD}\sim \frac{2a}{BR},
  \qquad\text{so}\qquad
  BD\sim \frac{A'B'\,BR}{2a}.
  \qquad\text{(3)}
\]
Also $A'B'\sim r\delta'$ and $BR=2r\delta'$, hence $A'B'\sim \tfrac12 BR$. Therefore
\[
  BD\sim \frac{(\tfrac12 BR)\,BR}{2a}=\frac{BR^2}{4a},
\]
so $\dfrac{BR^2}{BD}\to 4a$.

An equivalent shortcut (as suggested) is to use $\triangle SDB\sim\triangle FDS$, giving
\[
  DS^2\sim DB\cdot FS,\qquad DB\sim \frac{DS^2}{FS}.
\]
With $DS=\tfrac12 AD$ and the same $\alpha$-projection relations, one again obtains
\[
  DB\sim \frac{BR^2}{4a},
\]
hence the same limit $\dfrac{BR^2}{BD}\to 4a$.
\end{proof}

\Cref{prop:parabola-br2-over-bd} is the inverse-problem parabolic counterpart of \eqref{eq:proof2-bd-br2} in the elliptic setting. Since
\[
  \frac{BR^2}{BD}\to 4a
  \quad\Longleftrightarrow\quad
  \frac{BD}{BR^2}\to \frac{1}{4a},
\]
\Cref{lem:bd-over-br2} gives
\[
  a_F=\frac{2L^2}{r^2}\cdot\frac{1}{4a}
  =\frac{L^2}{2a}\,\frac{1}{r^2}
  =\frac{L^2}{p}\,\frac{1}{r^2}
  \qquad (p=2a),
\]
and therefore
\[
  \mathbf a_F(r)=-\mu\,\frac{\mathbf r}{r^3},
  \qquad
  \text{i.e. }\mathbf a_F(r)=-\frac{\mu}{r^2}\,\hat r,
  \qquad
  \mu=\frac{L^2}{p}
  \quad (p=2a).
\]

\subsubsection{Parabola variant of inverse problem Proof 3: direct hodograph-circle use}
\label{subsubsec:parabola-proof3}

\begin{proposition}[Parabola variant of inverse problem Proof~3]
Using the hodograph-circle quantities in \Cref{fig:parabola-proof2}, one again obtains
\[
  a_F=\frac{L^2}{p}\frac{1}{r^2},
  \qquad
  \mathbf a_F(r)=-\mu\,\frac{\mathbf r}{r^3},
  \qquad
  \mu=\frac{L^2}{p},
  \quad (p=2a).
\]
\end{proposition}

\begin{proof}
In the parabola construction, let $A_1,B_1$ be the corresponding points on the hodograph-related circle. Since $MA_1\perp FA$,
\[
  FA_1\cdot FA' = FM^2 = a^2.
\]
Using the areal-rate relation in the same form as before,
\[
  v\cdot FA' = L
  \quad\Longrightarrow\quad
  FA_1=\frac{v\,a^2}{L}.
\]
For a neighboring point $B$, this gives
\[
  A_1B_1=\Delta(FA_1)=\frac{a^2}{L}\,\Delta v.
\]
From the local parabola similarity in this construction,
\[
  DB=\frac{2r}{a}\,A_1B_1
  =\frac{2ra}{L}\,\Delta v.
\]
Also, from the same local area step,
\[
  DB\cdot r = L\,\Delta t.
\]
Hence
\[
  \frac{2ra}{L}\,\Delta v\cdot r=L\,\Delta t
  \quad\Longrightarrow\quad
  a_F=\frac{\Delta v}{\Delta t}
  =\frac{L^2}{2a}\frac{1}{r^2}
  =\frac{L^2}{p}\frac{1}{r^2},
  \qquad (p=2a).
\]
Therefore
\[
  \mathbf a_F(r)=-\mu\,\frac{\mathbf r}{r^3},
  \qquad
  \text{i.e. }\mathbf a_F(r)=-\frac{\mu}{r^2}\,\hat r,
  \qquad
  \mu=\frac{L^2}{p}.
\]
\end{proof}

\section{Forward Problem: Proof 1F}
\label{sec:forward-problem-hodograph}
\subsection{Historical Hodograph Background}

\subsubsection{Newton, uniqueness, and asserted conics}

From the \emph{Principia} viewpoint (1687), Newton gives a geometric proposition-chain establishing both directions: conic with force at a focus $\Rightarrow$ inverse-square law, and inverse-square centripetal attraction $\Rightarrow$ conic orbits (ellipse/parabola/hyperbola by regime) \cite{newton1846mottewikisource,chandrasekhar1995principia}. What he does \emph{not} provide is a modern initial-value existence/uniqueness theorem for the forward problem; the argument is synthetic and limit-geometric (ultimate ratios), not an ODE well-posedness proof.

In parallel, early Continental analysts (late 17th to early 18th century) reduced the central-force problem by $u(\theta)=1/r$, yielding Binet's equation. For $F(r)=-\mu/r^2$:
\begin{equation}
  u''(\theta)+u=\frac{\mu}{h^2},
\end{equation}
with solution
\begin{equation}
  u(\theta)=\frac{\mu}{h^2}\Bigl(1+e\cos(\theta-\theta_0)\Bigr),
\end{equation}
i.e. a conic in polar form. However, deriving this form by itself is not a substitute for a uniqueness theorem: without an existence/uniqueness framework, one has not yet formalized exclusion of other possible local branches/continuations. Historical accounts of this transition to differential methods are discussed by Nauenberg (2003) \cite{nauenberg2003keplerarea}.

So the modern statement ``initial data determine a unique orbit'' is best read as a reconstruction of Newton's practice rather than a theorem he states in contemporary form. The same caution applies to early Continental differential reductions: the orbit equation gives candidate solution families, while later uniqueness theory is what formally secures single-orbit determinacy from initial data. Likewise, hodograph language is absent in the \emph{Principia} (1687) and appears later with Hamilton (1847). Rigorous local existence/uniqueness theorems for ODEs were developed much later (Cauchy in the 1820s; then Lipschitz/Picard--Lindel\"of in the late 19th century, roughly the 1870s--1890s), i.e. more than a century after 1687.

\subsubsection{Hamilton's circular hodograph for the inverse-square law}

Hamilton introduced the hodograph in 1847 as the curve traced by the tip of the velocity vector $\mathbf v(t)$ in velocity space, and showed that for motion under a central inverse-square force the hodograph is a circle \cite{hamilton1847hodograph}. Our proofs use this as a key starting point. Without offering new insights to this part, we simply explain the modern proof and refer to Feynman's work for a discrete version of the same argument \cite{goodstein1996feynman}.

Write the dynamics as
\[
  \ddot{\mathbf r}=-\frac{\mu}{r^2}\,\hat{\mathbf r},
  \qquad
  L:=r^2\dot\theta=\text{const}.
\]
Reparametrize by $\theta$:
\begin{equation}
  \frac{d\mathbf v}{d\theta}
  =
  \frac{d\mathbf v}{dt}\frac{dt}{d\theta}
  =
  \ddot{\mathbf r}\,\frac{r^2}{L}
  =
  -\frac{\mu}{L}\,\hat{\mathbf r}.
\label{eq:hodo-step1}
\end{equation}
Using the polar-frame identity
\[
  \frac{d\hat{\boldsymbol\theta}}{d\theta}=-\hat{\mathbf r},
\]
equation~\eqref{eq:hodo-step1} becomes
\begin{equation}
  \frac{d\mathbf v}{d\theta}
  =
  \frac{\mu}{L}\frac{d\hat{\boldsymbol\theta}}{d\theta}.
\label{eq:hodo-step2}
\end{equation}
Integrating once gives
\begin{equation}
  \frac{d}{d\theta}\!\left(\mathbf v-\frac{\mu}{L}\hat{\boldsymbol\theta}\right)=0
  \quad\Longrightarrow\quad
  \mathbf v=\mathbf C+\frac{\mu}{L}\hat{\boldsymbol\theta},
\label{eq:hodo-step3}
\end{equation}
where $\mathbf C$ is a constant vector.
Because $\|\hat{\boldsymbol\theta}\|=1$,
\[
  \|\mathbf v-\mathbf C\|=\frac{\mu}{L}.
\]
Hence the hodograph (the tip of $\mathbf v$ in velocity space) is a circle of radius $\mu/L$, centered at $\mathbf C$. Feynman provided a discrete argument without these vector derivatives, based on essentially the same idea \cite{goodstein1996feynman}. Replacing the differential operator $d$ by the finite-difference operator $\Delta$ gives
\begin{equation}
  \Delta\mathbf v = \frac{\mu}{L}\,\Delta\hat{\boldsymbol \theta}.
\label{eq:hodo-step1-discrete}
\end{equation}
This describes a constant-curvature curve in velocity space: a circle of radius $\mu/L$ whose center corresponds to the force center, since both $\Delta\mathbf{v}$ and $\Delta\hat{\boldsymbol\theta}$ lie in the radial direction. Feynman gives a clear presentation of this discrete argument \cite{goodstein1996feynman}.

This is the converse of what we showed in previous sections, where we established that conic motion with constant areal speed (equivalently, motion under a central force) is a sufficient condition for the hodograph to be a circle. Historically, the circular hodograph has become a central tool for both the inverse and forward problems \cite{maxwell1877mattermotion,derbes2001reinventing,goodstein1996feynman,carinena2016newlook}. Later in this section, we show how the forward problem can be solved by two methods.

\subsubsection{Feynman's Lost Lecture, strengths, and later repairs}

Feynman's 1964 lecture re-popularized the hodograph method for teaching: first establishing a circular hodograph for inverse-square attraction, then reconstructing the conic orbit geometrically \cite{goodstein1996feynman}. Later work, including Derbes (2001), clarified missing details in tangent placement and invariants; in particular, angular momentum conservation (or the centripedal nature of the force) is the extra ingredient that fixes line placement and scale \cite{derbes2001reinventing}. More recent refinements, including van Haandel--Heckman (2009) and follow-up work by Cari\~nena--Ra\~nada--Santander (2016), may be read as using either (i) the directrix circle centered at the second focus, or (ii) an equivalent directrix-circle/hodograph picture centered at the force center, with a $\pm 90^\circ$ rotation between velocity and geometric proxy \cite{vanhaandelheckman2009teaching,carinena2016newlook}. CRS16 emphasizes that the force-center version gives a cleaner Feynman-style transformation.

In our constructions (\Cref{fig:hyperbola,fig:parabola-proof2}), the auxiliary circle has the same dual role. In particular, from the local identity used in \Cref{lem:product-identity},
\[
  FH\cdot FK' = b^2,
  \qquad
  F'\!A' = FK',
\]
so $\overrightarrow{F'\!A'}$ and $\overrightarrow{FK'}$ are colinear with opposite orientation and equal magnitude. Hence one can use the auxiliary-circle data with either center choice ($F$ or $F'$) as a hodograph proxy, corresponding to a $+90^\circ$ or $-90^\circ$ rotation. For the parabola, four equivalent rotated-hodograph constructions are shown in \cref{fig:parabola-proof2}; Derbes' choice is one of them (his $C_4$ circle). Each provides a different geometric perspective. The remaining step is the energy classification: $E<0$ gives an ellipse, $E=0$ the parabolic limit, and $E>0$ a hyperbola, equivalently with the focus lying inside, on, or outside the hodograph circle, respectively.

\subsection{Infinitesimal Hodograph Construction for the Forward Problem}

We take as given (from the historical discussion above) the hodograph-circle fact: for planar motion under a central inverse-square force, the hodograph of $\mathbf v(t)$ is a circle in velocity space.

Our goal is to recover the orbit form $r(\theta)$ (forward problem) by a finite-step argument aligned with Newton's impulse-triangle logic, minimizing reliance on continuous calculus.

In the next section, inspired by the geometric methods above, we give a forward-proof route based on explicit geometric construction. In the present section, however, we first take a slightly different approach: an infinitesimal hodograph analysis (Proof~1F).

\subsubsection{Instantaneous view of Hodograph as Proof 1F}

Assume planar motion under a centripetal acceleration toward $F$:
\begin{equation}
  a=\frac{\mu}{r^2},
  \qquad
  r=FA,
\end{equation}
and constant specific angular momentum
\begin{equation}
  L=r^2\dot\theta.
\end{equation}
Fix a small time step $\Delta t$, with consecutive points $B\to A$ on the orbit, $BR\perp FA$, and the translated $\Delta t$-scaled hodograph geometry in which
\begin{equation}
  OM=BR+OR'.
\end{equation}
Define
\begin{equation}
  u:=\frac{\mu}{L}.
\end{equation}

\begin{figure}[t]
  \centering
  \includegraphics[width=0.95\linewidth]{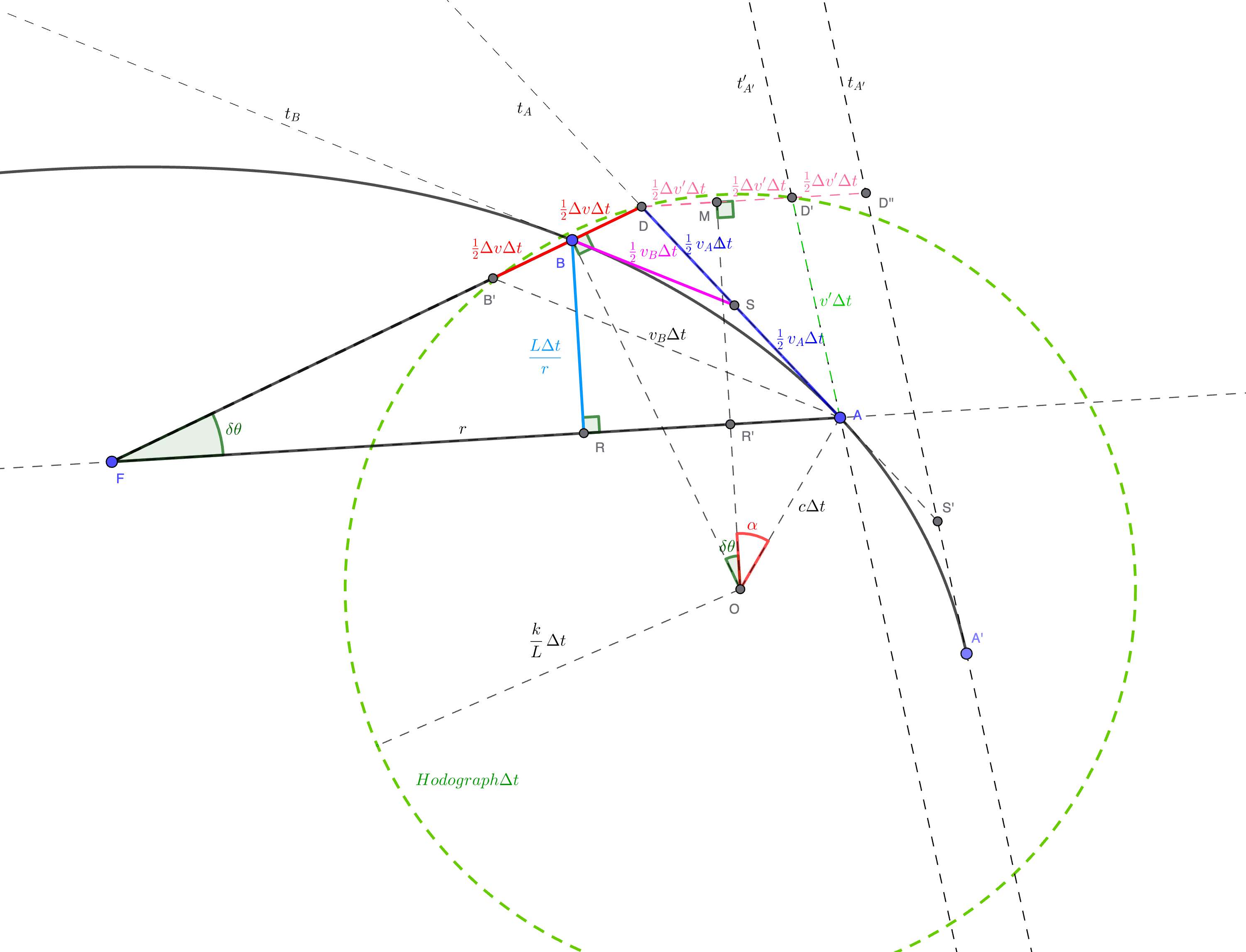}
  \caption{Forward-problem construction: previous/current/next velocities ($V_{A'},V_A,V_B$), each scaled by $\Delta t$, are shown on the translated $\Delta t$-scaled hodograph circle, with decomposition $OM=BR+OR'$. As $\Delta t\to 0$, the hodograph circle and all lengths proportional to $\Delta t$ shrink accordingly, while the deflection segments proportional to $\Delta v\,\Delta t$ shrink faster, all to an infinitesimally small circle near point A; the construction is therefore the limit of discrete approximations.}
  \label{fig:forward-problem}
\end{figure}

\begin{proposition}[Hodograph-radius analysis $\Rightarrow$ conic form]
\label[proposition]{prop:inverse-hodograph-proof}
With the notation of \Cref{fig:forward-problem}, the hodograph of $\mathbf v$ is a circle of radius $u$ (hence the $\Delta t$-scaled hodograph circle has radius $u\Delta t$). If the fixed shift is written directly in eccentricity form as
\begin{equation}
  OR'=e\,u\,\Delta t\cos\alpha,
\end{equation}
then
\begin{equation}
  \frac1r=\frac{\mu}{L^2}\left(1-e\cos\alpha\right)
  \quad\Longleftrightarrow\quad
  r(\alpha)=\frac{p}{1-e\cos\alpha},
  \qquad
  p=\frac{L^2}{\mu}.
\label{eq:proof1f-r-alpha}
\end{equation}
For a direct geometric interpretation of this polar form in the ellipse construction, see the discussion following \Cref{subsec:ellipse-proof2f}, especially \Cref{eq:ellipse-vn-polar-geom}.
Moreover,
\begin{equation}
  \begin{aligned}
    e<1,\;=1,\;>1
    &\Longleftrightarrow
    \text{velocity-origin inside/on/outside hodograph circle}\\
    &\Longleftrightarrow
    \text{ellipse/parabola/hyperbola}.
  \end{aligned}
\label{eq:proof1f-regimes}
\end{equation}
\end{proposition}

\begin{proof}
In \Cref{fig:forward-problem}, the previous, current, and next velocities ($V_{A'},V_A,V_B$) are placed together in velocity space after scaling by $\Delta t$; this is exactly the translated $\Delta t$-scaled hodograph circle used in the relations below.

First, $\Delta v$ depends only on $\Delta\theta$.
Over a short interval $\Delta t$,
\begin{equation}
  \Delta v=a\,\Delta t=\frac{\mu}{r^2}\Delta t.
\end{equation}
From $L=r^2\dot\theta$, one has $\Delta t=\frac{r^2}{L}\Delta\theta$. Hence
\begin{equation}
  \Delta v=\frac{\mu}{r^2}\cdot\frac{r^2}{L}\Delta\theta=\frac{\mu}{L}\Delta\theta=u\,\Delta\theta,
  \qquad
  \frac{\Delta v}{\Delta\theta}=u\ \text{(constant)}.
\label{eq:inverse-step1}
\end{equation}

Next, conclude the hodograph-circle radius.
Because the force is central, $\Delta\mathbf v$ is radial inward. As $FA$ turns by $\Delta\theta$, the direction of $\Delta\mathbf v$ turns by the same $\Delta\theta$, while its magnitude is $u\,\Delta\theta$ by \eqref{eq:inverse-step1}. Thus the velocity-tip polygon is the equal-angle, equal-arc-length limit of an inscribed circle polygon, so the hodograph is a circle of radius $u$. Multiplying by $\Delta t$ gives
\begin{equation}
  OM=u\,\Delta t.
\label{eq:inverse-step2}
\end{equation}

Then, use tangential displacement.
The tangential speed is $v_\theta=L/r$, hence
\begin{equation}
  BR=v_\theta\Delta t=\frac{L}{r}\Delta t.
\label{eq:inverse-step3}
\end{equation}

Now use the key geometry in the plot.
From the translated-circle construction, $OM=MR'+OR'$ and $MR'=BR$, so
\begin{equation}
  OM=BR+OR'.
\label{eq:inverse-step4}
\end{equation}
Assume
\begin{equation}
  OR'=e\,u\,\Delta t\cos\alpha.
\label{eq:inverse-step5}
\end{equation}
Substitute \eqref{eq:inverse-step2}, \eqref{eq:inverse-step3}, \eqref{eq:inverse-step5} into \eqref{eq:inverse-step4}:
\begin{equation}
  u\Delta t=\frac{L}{r}\Delta t+e\,u\,\Delta t\cos\alpha.
\end{equation}
Cancel $\Delta t$:
\begin{equation}
  u=\frac{L}{r}+e\,u\cos\alpha
  \quad\Longrightarrow\quad
  \frac1r=\frac{u}{L}\left(1-e\cos\alpha\right)
  =\frac{\mu}{L^2}\left(1-e\cos\alpha\right).
\end{equation}
This is equivalent to
\begin{equation}
  r(\alpha)=\frac{p}{1-e\cos\alpha},
  \qquad
  p=\frac{L^2}{\mu}.
\end{equation}

Finally, classify by the origin position in velocity space.
The hodograph circle has radius $u$, and the velocity-origin offset has magnitude $eu$. Hence the origin is inside/on/outside the hodograph circle according as $eu<u$, $eu=u$, $eu>u$, i.e. $e<1,=1,>1$. These are precisely the ellipse/parabola/hyperbola regimes.
\end{proof}
This concludes the first forward-problem proof, which is a direct infinitesimal version of the hodograph construction. In the next section, the second proof follows a geometric insight similar to Feynman's Lost Lecture method, with stronger emphasis on the auxiliary circle as a proxy for the hodograph.

% !TEX root = ../../main.tex
\section{Forward Problem: Proof 2F}
\label{sec:forward-problem-proof2f}

We again start from the forward-problem assumption: a centripetal inverse-square law $a_F\propto 1/r^2$.
By Section~\ref{sec:forward-problem-hodograph}, this implies a circular hodograph in velocity space.
This section presents a second route, using that hodograph-circle nature directly through rotated/scaled auxiliary circles.
As in earlier sections, we intentionally over-scribe assisting lines so readers can spot more angle and length relations at a glance.
For each conic, a simpler base shape is enough for a minimal proof, but the enriched diagram exposes multiple valid relation chains.
In the subsections below we highlight several such chains so readers can develop alternative proof techniques from the same geometric idea.
Throughout this section, whenever the scale factors \(b^2/L\) (ellipse/hyperbola) or \(a^2/L\) (parabola, with \(a:=MF\)) appear, see \Cref{subsec:proof2f-discussion} for their equivalent expressions purely in initial data \((\mu,L,r_0,d_0)\). In the same notation, \(b^2/p=a\) and \(a^2/p=p/4\), so these geometric ratios are also determined by initial conditions.

\subsection{Ellipse: Forward Problem Proof 2F}
\label{subsec:ellipse-proof2f}

\begin{figure}[t]
  \centering
  \includegraphics[width=\linewidth]{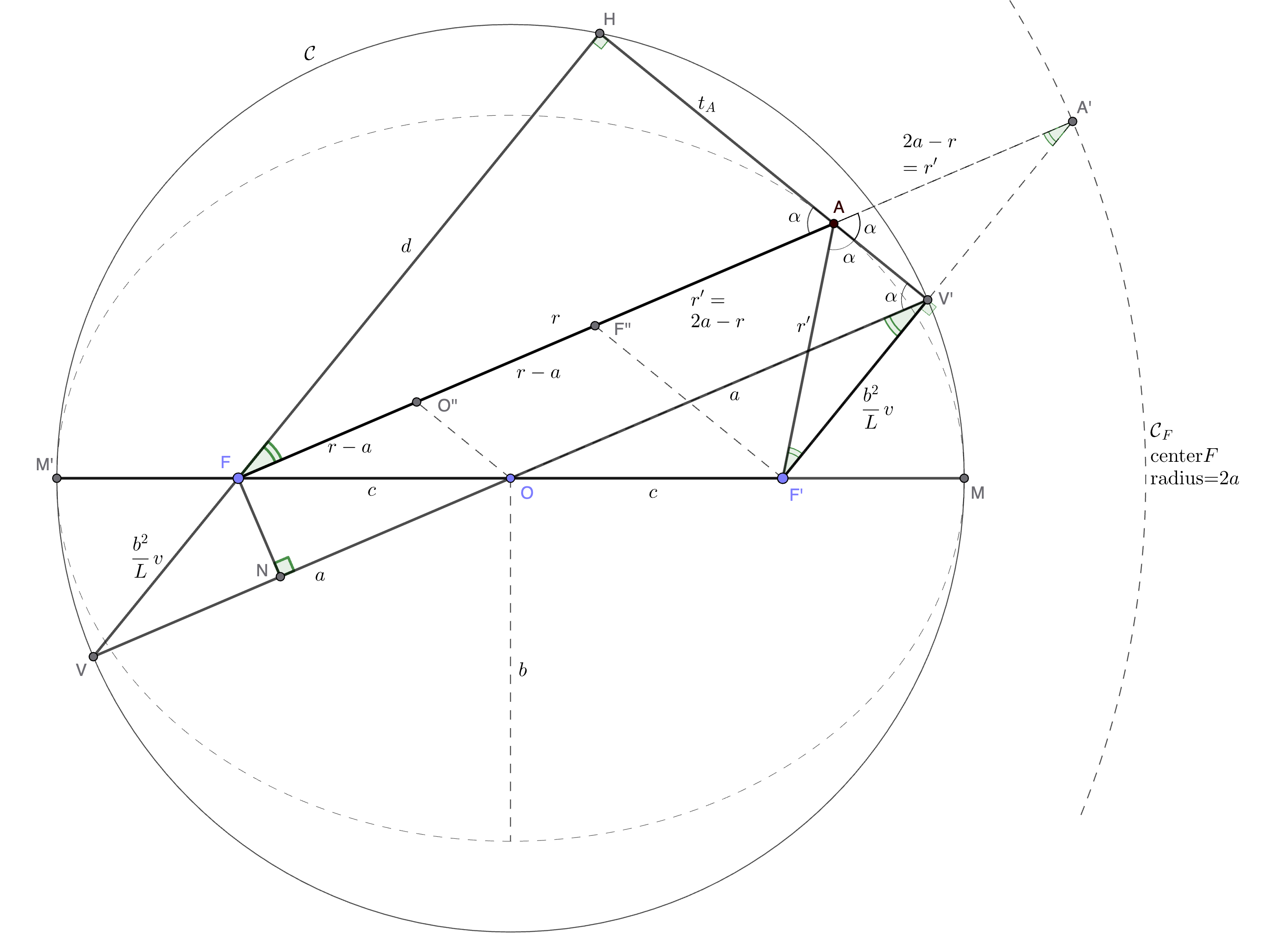}
  \caption{Ellipse configuration for Forward Problem Proof~2F.}
  \label{fig:ellipse-proof2f}
\end{figure}

Using \Cref{fig:ellipse-proof2f}, we set up Forward Problem Proof~2F for the ellipse case in the circle-to-orbit framework.
In this subsection we also use the standard elliptic identity for the semi-latus rectum \(p\): geometrically, \(p\) is the perpendicular distance from the force center \(F\) to the orbit along the line through \(F\) orthogonal to the major axis. In \((a,b,c)\) notation,
\begin{equation}
  p^2+(2c)^2=(2a-p)^2
  \quad\Longrightarrow\quad
  p=\frac{a^2-c^2}{a}=\frac{b^2}{a}.
\label{eq:p-b2-over-a-ellipse}
\end{equation}

\begin{proposition}[Forward Problem Proof~2F: ellipse case]
Assume inverse-square centripetal attraction, so the hodograph is a circle. If the velocity origin lies inside that hodograph circle, then the orbit is an ellipse.
\end{proposition}

\begin{proof}
Start with the hodograph circle of radius $u$. Rotate it counterclockwise by $90^\circ$ and scale by \(\beta/L\), where \(\beta\) is a length-squared parameter left undetermined at this stage, and \(L\) is the specific angular momentum (twice areal speed, fixed by central-force motion from Section~2; later \(\beta\) is fixed to \(b^2\), and \(a,b,c\) are expressed from initial data in \Cref{eq:proof2f-disc-shodo-u-E0,eq:proof2f-disc-abc-from-initial}). Denote the scaled circle by $\mathcal C$ with center $O$, radius $a$, and
\[
  OF=:c,
  \qquad
  a=\frac{\beta}{L}u.
\]
Let the line $OF$ meet $\mathcal C$ at $M'$ (near $F$) and $M$ (far from $F$). For any point $V\in\mathcal C$, let the line $VF$ meet $\mathcal C$ again at $H$.

By the Euclidean power-of-a-point theorem (secant-secant form),
\[
  FM'\cdot FM = FV\cdot FH = (a-c)(a+c)=a^2-c^2.
\]
Write $FH=:d$ and, from the chosen scaling, $FV=\frac{\beta}{L}v$; then
\[
  a^2-c^2=\frac{\beta}{L}\,v\,d.
\]
Now choose the normalization
\[
  \beta:=a^2-c^2.
\]
For the ellipse case we denote this positive quantity by
\[
  b^2:=a^2-c^2.
\]
Hence $v\,d=L$. Therefore, the point where the true velocity line (perpendicular to $FH$) meets the tangent line $t_A$ is exactly the constructed point $H$.

Now draw $t_A$ through $H$ perpendicular to $FH$, and let $V'$ be the second intersection of $t_A$ with $\mathcal C$. Since $\angle VHV'=90^\circ$, $VV'$ is a diameter of $\mathcal C$, so $O\in VV'$. Draw through $V'$ the line perpendicular to $HV'$ and let it meet $OF$ at $F'$. Then $F'V'\parallel VH$, and the corresponding triangles give
\[
  F'V'=FV=\frac{b^2}{L}v,
  \qquad
  OF'=OF=c.
\]
So $F'$ is uniquely determined as the reflection of $F$ across $O$, independent of the choice of $V$.

To locate the orbital point $A$ on $t_A$, use direction information: $OV$ is normal to $\mathcal C$ at $V$, hence normal to the local tangent proxy, so it is parallel to the acceleration direction. Because the force is centripetal, this gives
\[
  AF\parallel V'V.
\]
Let $A'=FA\cap F'V'$. Since $AF\parallel V'V$ and $O\in V'V$, we have $OV'\parallel FA'$. In $\triangle FF'A'$, $O$ is the midpoint of $FF'$ and $OV'\parallel FA'$, so $V'$ is the midpoint of $F'A'$. Therefore
\[
  FA'=2\,OV'=2a.
\]
Also, $A,V'\in t_A$, while $A'F'$ is collinear with $F'V'$ and $F'V'\perp t_A$; hence $AV'\perp A'F'$. Since $V'$ is the midpoint of $A'F'$, the line $AV'$ is the perpendicular bisector of $A'F'$. Therefore
\[
  AA'=AF',
  \qquad
  \angle V'AF'=\angle V'AA'.
\]
This is the reflection characterization at $A$ across the tangent line $t_A$. Combining $AA'=AF'$ with the distance relation above gives
\[
  AF+AF' = AF+AA' = FA' = 2a.
\]
This is the gardener characterization of an ellipse (constant sum of distances to two fixed foci), with foci $F,F'$. Therefore, when the velocity origin is inside the hodograph circle, the orbit is elliptic.
\end{proof}

\paragraph{Alternative closing step without constructing \(A'\): ellipse proof 2F'.}
\begin{proof}[Alternative proof (ellipse 2F')]
One may branch from the step \(AF\parallel V'V\) (just before introducing \(A'\)). Define
\[
  r:=FA,
  \qquad
  r':=F'A.
\]
Draw \(F'F''\parallel t_A\) and let it meet \(FA\) at \(F''\). Draw \(OO''\parallel t_A\) and let it meet \(FA\) at \(O''\) (on the branch from \(F\) to \(A\)).

Since \(O\in V'V\), the relation \(AF\parallel V'V\) implies \(AF\parallel OV'\). Also \(OO''\parallel t_A\parallel AV'\). Hence \(O''AV'O\) is a parallelogram, so
\[
  O''A=OV'=a,
\]
where \(OV'=a\) is the radius definition of \(\mathcal C\). Therefore
\[
  FO''=FA-O''A=r-a.
\]
In \(\triangle FF'F''\), \(O\) is the midpoint of \(FF'\), and \(OO''\parallel F'F''\); hence \(O''\) is the midpoint of \(FF''\), so
\[
  FF''=2\,FO''=2(r-a).
\]
Thus
\[
  AF''=AF-FF''=r-2(r-a)=2a-r.
\]

The reflection-angle relation at \(A\) used above can then be established via the right-triangle similarity \(\triangle AHF\sim\triangle AV'F'\), since \(V'F':FH=VF:FH=AV':AH\). Set
\[
  \angle FAH=\angle F'AV'=: \alpha.
\]
Because \(F'F''\parallel t_A\) and \(AV'\subset t_A\),
\[
  \angle AF''F'=\angle FAH=\alpha,
  \qquad
  \angle AF'F''=\angle F'AV'=\alpha.
\]
Hence \(\triangle AF'F''\) is isosceles, so \(AF''=AF'=r'\). Combining with \(AF''=2a-r\), we get
\[
  r'=2a-r
  \quad\Longrightarrow\quad
  r+r'=2a.
\]
This is exactly the gardener characterization of an ellipse, completing the same conclusion without constructing \(A'\), i.e., the alternative ellipse Proof~2F'.
\end{proof}

\paragraph{Alternative polar-form closure (ellipse Proof~2F'').}
\begin{proof}[Alternative proof (ellipse 2F''): polar-form closure]
As an alternative closure of \Cref{subsec:ellipse-proof2f}, we verify directly that the constructed orbit obeys the ellipse polar formula (\Cref{eq:proof1f-r-alpha}). Let \(V\) denote the point on \(\mathcal C\) whose radius satisfies \(OV\parallel AF\) (in the present construction). Drop the perpendicular from \(F\) to the line \(OV\), and call the foot \(N\). Then \(N\in OV\), so
\[
  VN=OV-ON=a-c\cos\angle OFA.
\]

Now use \(\triangle VNF\sim\triangle AHF\), hence
\[
  \frac{VF}{FA}=\frac{VN}{FH}
  \quad\Longrightarrow\quad
  VF\cdot FH=VN\cdot FA.
\]
With \(FA=r\) and \(VF\cdot FH=b^2\), this gives \(VN=b^2/r\). Therefore
\begin{equation}
  \frac{b^2}{r}=a-c\cos\angle OFA=VN.
\label{eq:ellipse-vn-polar-geom}
\end{equation}
Equivalently,
\[
  r=\frac{b^2}{a-c\cos\angle OFA}.
\]
This is the geometric meaning of the polar ellipse relation in this construction. Analogous identities exist for hyperbola and parabola, but we do not detail them here.
\end{proof}

At this stage, the orbit is parameterized geometrically by \((a,b,c)\). Its physical determination, however, is by dynamical data: the force-law constant and the initial state \((\mathbf r_0,\mathbf v_0)\) (equivalently \(r_0,v_0\) with direction, or \(r_0,d_0,v_0\)). We defer this equivalence intentionally to \Cref{subsec:proof2f-discussion}. The reason is organizational: we first complete the same hodograph/force-center construction logic for hyperbola and parabola in \Cref{subsec:hyperbola-proof2f,subsec:parabola-proof2f}, then present one unified parameter map for all cases in \Cref{subsec:proof2f-discussion}.

\subsection{Hyperbola: Forward Problem Proof 2F}
\label{subsec:hyperbola-proof2f}

\begin{figure}[t]
  \centering
  \includegraphics[width=\linewidth]{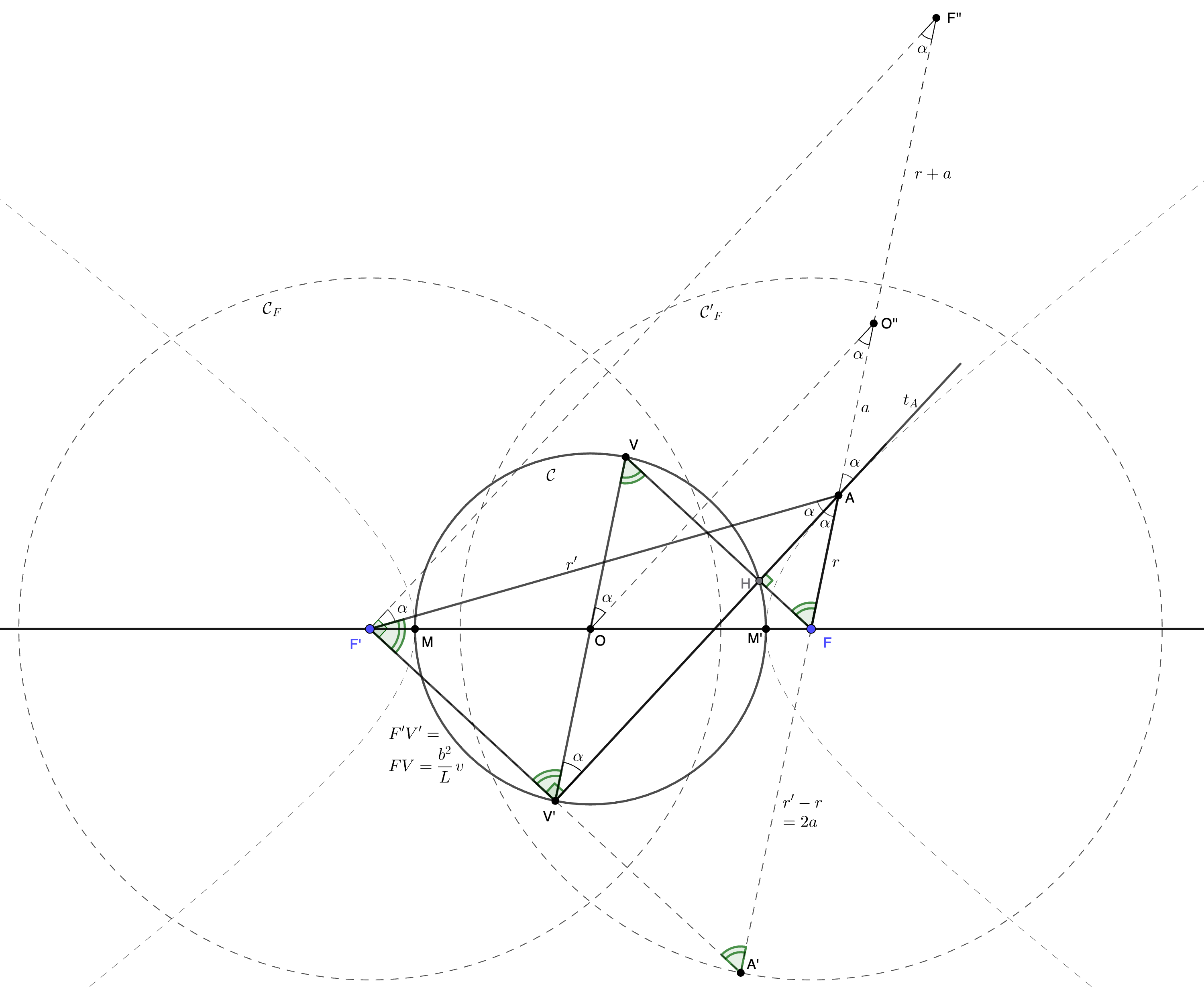}
  \caption{Hyperbola configuration for Forward Problem Proof~2F.}
  \label{fig:hyperbola-proof2f}
\end{figure}

Using \Cref{fig:hyperbola-proof2f}, we set up the hyperbolic counterpart with the same hodograph-circle logic and branch-sign adjustments.

\begin{proposition}[Forward Problem Proof~2F: hyperbola case]
Assume inverse-square centripetal attraction, so the hodograph is a circle. If the velocity origin lies outside that hodograph circle, then the orbit is a hyperbola.
\end{proposition}

\begin{proof}
This proof is almost the same as the elliptic proof in \Cref{subsec:ellipse-proof2f}; the same is true for the alternative proof below. Using the same \(\beta/L\)-scaling template, the only essential normalization change is
\[
  \beta=b^2=c^2-a^2,
\]
which is the hyperbolic identity that matches the rotated/scaled hodograph circle \(\mathcal C\) to the orbit geometry (and the same \(b^2/L\) scale can be rewritten from initial data as in \Cref{subsec:proof2f-discussion}).

With \(OF=c\), \(FV=\frac{b^2}{L}v\), and \(FH=d\), the secant product gives
\[
  FM'\cdot FM=FV\cdot FH=(c-a)(c+a)=c^2-a^2=b^2,
\]
hence \(vd=L\). So the tangent-placement step and the constructions of \(V'\) and \(F'\) proceed exactly as in the ellipse case.

Now set \(r:=AF\) and \(r':=AF'\), and define \(A':=FA\cap F'V'\). The midpoint argument in \(\triangle FF'A'\) is unchanged, yielding
\[
  FA'=2a=V'V.
\]
The reflection property at \(A\) is also unchanged, so \(AA'=AF'=r'\). On the chosen branch (\(A,F,A'\) collinear with \(F\) between \(A\) and \(A'\)),
\[
  r'-r = AA'-AF = FA' = 2a.
\]
This is the gardener characterization of a hyperbola (\(\lvert AF'-AF\rvert=2a\)) with foci \(F,F'\). The remaining line-by-line details are parallel to \Cref{subsec:ellipse-proof2f}.
\end{proof}

\paragraph{Alternative closing step (hyperbola Proof~2F').}
\begin{proof}[Alternative proof (hyperbola 2F')]
The alternative closure from ellipse Proof~2F' transfers directly, with one branch change: construct \(F''\) and \(O''\) on the ray \(FA\) beyond \(A\) (rather than on the segment \(FA\)). Draw \(F'F''\parallel t_A\) and \(OO''\parallel t_A\), meeting the ray \(FA\) at \(F''\) and \(O''\), respectively.

Then \(O''A=OV'=a\), so
\[
  FO''=FA+AO''=r+a.
\]
As before, \(O\) is the midpoint of \(FF'\), and \(OO''\parallel F'F''\), so \(O''\) is the midpoint of \(FF''\). Hence
\[
  FF''=2(r+a),
  \qquad
  AF''=FF''-AF=r+2a.
\]
The same reflection-angle argument gives \(AF''=AF'=r'\). Therefore
\[
  r'=r+2a
  \quad\Longleftrightarrow\quad
  r'-r=2a.
\]
Again, the remaining details are routine parallels of the elliptic case.
\end{proof}

\paragraph{Alternative polar-form closure (hyperbola Proof~2F'').}
\begin{proof}[Alternative proof (hyperbola 2F''): polar-form closure]
This is the direct hyperbolic counterpart of ellipse Proof~2F'' in \Cref{subsec:ellipse-proof2f}. Draw the perpendicular from \(F\) to the line \(VO\), and call the foot \(N\); then use the segment \(VN\). With the same similarity step (\(\triangle VNF\sim\triangle AHF\)) and the same secant-product step, the derivation translates letter-for-letter from the ellipse case, with the hyperbolic normalization \(b^2=c^2-a^2\) and directed-segment sign convention on the chosen branch. Hence the same polar-form closure follows in this notation.
\end{proof}

\paragraph{Swapped-focus centrifugal variant.}
\begin{proof}[Swapped-focus centrifugal variant]
The centrifugal-force version follows almost exactly as well. In \Cref{fig:hyperbola-proof2f}, swap the names \(F\leftrightarrow F'\), rename the current \(V'\) as \(H\), and let the line through the new \(F'\) and \(H\) meet \(\mathcal C\) again at the new \(V\). After this relabeling, the rest of the construction and proof chain is sequentially identical, so we omit repetitive details.
\end{proof}

As in \Cref{subsec:ellipse-proof2f}, this subsection keeps the geometry-first derivation; the equivalent initial-data determination of the scaling and conic parameters is deferred to \Cref{subsec:proof2f-discussion}.

\subsection{Parabola: Forward Problem Proof 2F}
\label{subsec:parabola-proof2f}

\begin{figure}[t]
  \centering
  \includegraphics[width=\linewidth]{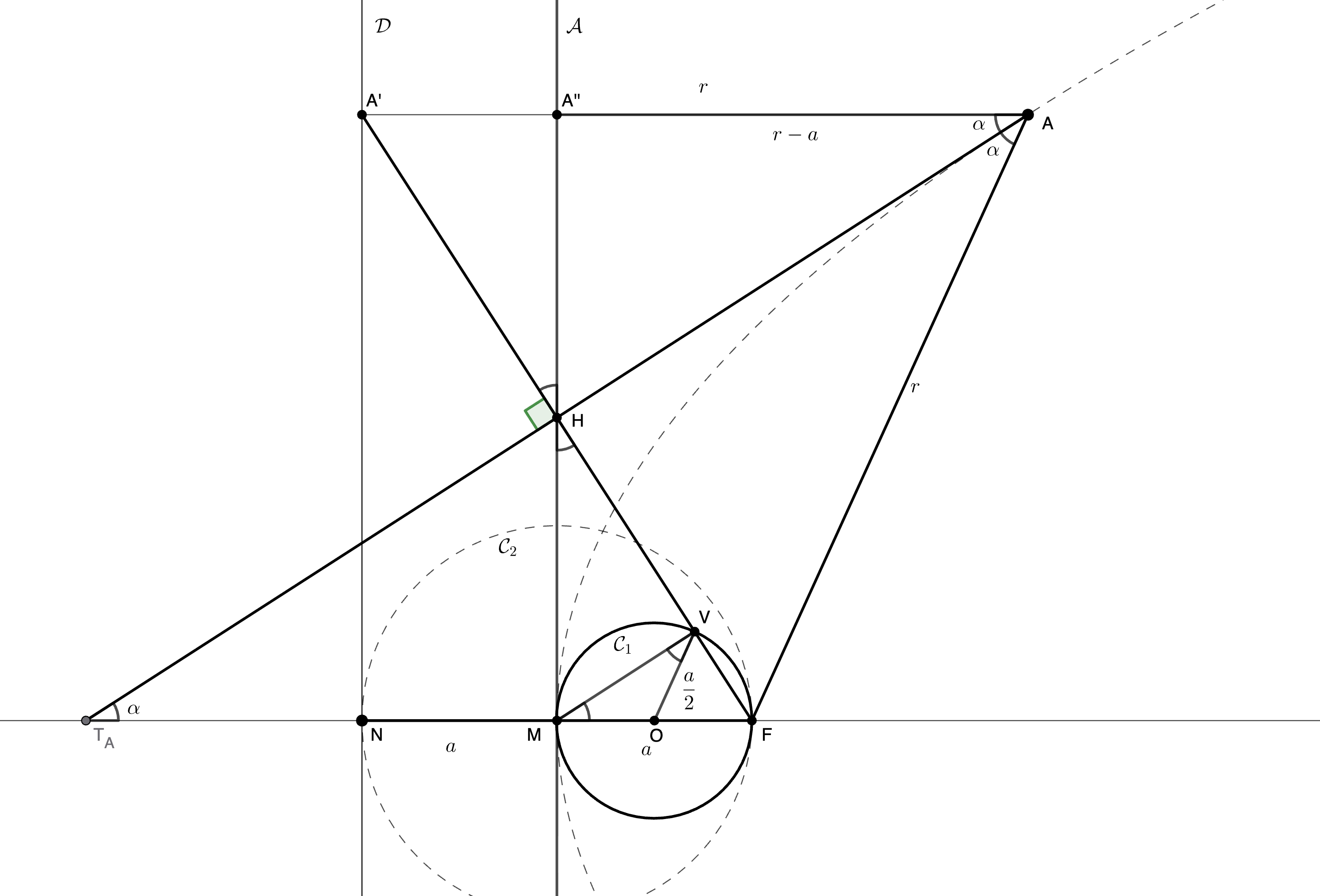}
  \caption{Parabola configuration for Forward Problem Proof~2F.}
  \label{fig:parabola-proof2f}
\end{figure}

Using \Cref{fig:parabola-proof2f}, we set up the parabolic threshold case in the same proof family.

\begin{proposition}[Forward Problem Proof~2F: parabola case]
Assume inverse-square centripetal attraction, so the hodograph is a circle. If the velocity origin lies on that hodograph circle, then the orbit is a parabola.
\end{proposition}

\begin{proof}
As in the previous two subsections, rotate the hodograph by \(90^\circ\), but now scale by \(a^2/L\) (with equivalent initial-data form in \Cref{subsec:proof2f-discussion}). In the parabola configuration of \Cref{fig:parabola-proof2f}, the corresponding point \(H\) lies on the fixed line \(\mathcal A\), where \(\mathcal A\perp FM\) and \(FM\) is the diameter of the scaled hodograph circle \(\mathcal C_1\).

For \(V\in\mathcal C_1\), let \(H=VF\cap\mathcal A\). Since \(\triangle MVF\) and \(\triangle HMF\) are right triangles with a common acute angle at \(F\) (because \(V,H,F\) are collinear), they are similar. Hence
\[
  FV\cdot FH=MF^2=a^2.
\]
Using the areal relation \(v\cdot FH=L\), we obtain
\[
  FV=\frac{a^2}{L}\,v,
\]
which is the parabolic analogue of the elliptic/hyperbolic scaling formulas.

Now draw \(t_A\) through \(H\) perpendicular to \(FH\). The orbital point for this \(V\) is
\[
  A:=t_A\cap\{\text{line through }F\text{ parallel to }OV\},
\]
from the same centripetal-direction argument used above.

Let \(T_A:=t_A\cap \overrightarrow{FM}\). Since \(MV\perp VF\) and \(t_A\perp VF\), we have \(MV\parallel AT_A\). Also \(OV\parallel AF\) and \(MO\parallel FT_A\), so \(\triangle A T_A F\sim\triangle VMO\). Therefore the tangent-bisector (reflection) relation is the same as before: the tangent at \(A\) bisects the angle between the focus ray \(AF\) and the fixed direction parallel to \(FM\).

For the directrix closure, draw through \(A\) the line perpendicular to \(\mathcal A\), and let it meet \(\mathcal A\) at \(A''\). By definition \(AF=r\). With the same parallel-line relations,
\[
  AA''=T_A M=F T_A-F M=F A-a=r-a.
\]
Now let \(\mathcal D\parallel\mathcal A\) meet the ray \(FM\) at \(N\) with \(NM=a\). Define \(A'\) as the intersection of ray \(AA''\) with \(\mathcal D\). Then
\[
  AA'=AA''+a=(r-a)+a=r=AF.
\]
Thus the distance to the focus equals the distance to the directrix, which is exactly the defining property of a parabola. (In the plotted construction, \(A',H,F\) are collinear; equivalently, one may define \(A'\) via \(FH\cap\mathcal D\) and verify \(AA'\parallel FM\).) The remaining routine checks are left to the reader.
\end{proof}

\paragraph{Alternative scaled-circle variant.}
\begin{proof}[Alternative scaled-circle variant]
The same conclusion can also be constructed with the auxiliary circle \(\mathcal C_2\) shown in \Cref{fig:parabola-proof2f}, corresponding to a \(2a^2/L\) hodograph scaling (again reducible to initial data as in \Cref{subsec:proof2f-discussion}). As in the ellipse and hyperbola cases, choosing a different but equivalent hodograph proxy changes only intermediate geometric details, not the final parabolic conclusion.
\end{proof}

Alternative polar-form closure (parabola Proof~2F'') in parabola is trivial and equivalent to the directrix closure. The same distance relation \(r=r\cos\beta+2a\) can be read directly from the diagram, where \(\beta\) is the angle between \(AF\) and the directrix direction \(FM\). This is the parabolic polar formula in this construction, with the same \(a^2/L\) scaling as before.

Again, we postpone the explicit initial-condition parameter map to \Cref{subsec:proof2f-discussion}: this subsection first establishes the geometric mechanism that forces the parabolic form.

\subsection{Discussion}
\label{subsec:proof2f-discussion}

\subsubsection{Rotation/Scaling Convention and Proxy Choice}

This subsection takes CRS as the primary reference point for rotation/scaling conventions in hodograph-based constructions \cite{carinena2016newlook}.

Before parameterizing by initial data, we clarify notation and conventions. Our Forward-Proof~2F scaling is not written exactly in the same normalization style as CRS. In CRS/CNS-style presentation, the scaling is often written as a signed factor tied to the energy sign (\(E<0\), \(E=0\), \(E>0\)), so negative scaling is allowed. Here we do not adopt that convention: we keep the geometric proxy radius positive (ellipse auxiliary circle \(\mathcal C\) has radius \(a\)) and encode orientation by rotation choice. The two conventions are equivalent and give the same geometric/dynamical results.

As emphasized in \cite{gonzalezvillanueva1998circular}, one can even avoid this preprocessing entirely: no mandatory rescaling of the hodograph and no mandatory shift of the velocity origin to specially convenient centers (such as one of the foci) are required in principle, because perpendicular-and-parallel constructions can still recover the orbit. We do not emphasize that route here, since carefully chosen hodograph proxies in configuration space make the Euclidean argument shorter and more transparent. So these choices are not necessities, but geometric conveniences chosen for proof design. Nonetheless, for readers new to the topic, we find the auxiliary circle the most natural primary proxy; to our knowledge, using the auxiliary circle explicitly as the primary hodograph proxy has not been proposed in the literature, and we regard this framing as one contribution of this work.

For ellipse, we take as the main convention a \(+90^\circ\) rotation in the same sense as orbital motion, followed by scaling to the auxiliary circle. An equivalent but less natural alternative is the \(-90^\circ\) transform with the empty focus convention. In the directrix-circle picture (radius \(2a\)), these two conventions appear as two corresponding transformed circles. This is one reason the auxiliary-circle proxy is algebraically cleaner in configuration space.

For hyperbola, the same duality appears with center roles reversed: in our main convention, \(+90^\circ\) gives the empty-focus-centered proxy, while \(-90^\circ\) gives the force-center-shared speed-origin convention. This can appear as a reversed rotation sign relative to the hyperbola convention in \cite{carinena2016newlook}; the difference is convention-level, since that treatment allows negative scaling.

For parabola, \(\Cref{fig:parabola-proof2}\) already exhibits multiple equivalent rotated-hodograph proxies. In the notation there, \(C_1,C_2,C_3,C_4\) may be read as \(-90^\circ\) proxy choices (with representative centers \(F,F,H,J\)); the dual \(+90^\circ\) reading corresponds to speed-origin placements \(M,N,M,N\), respectively. The figure does not explicitly draw all scaled velocity vectors for this second reading, but it leads to equivalent proofs with different intermediate Euclidean chains. Their radii \((a/2,\ a,\ a,\ 2a)\) produce different \(s_{\mathrm{hodo}}\) values. In this paper we take \(C_1\) as the main auxiliary proxy.

\subsubsection{Initial-Data Parameterization}

Across the ellipse, hyperbola, and parabola cases, Proof~2F and its variants follow one common four-step template. First, rotate and scale the hodograph circle to obtain a geometric proxy in configuration space (with conic-dependent scale choices). Second, combine the secant/product identity with the areal invariant \(L\) so that the tangent carrier \(t_A\) is fixed by the constructed point \(H\). Third, recover the force-center direction by imposing the parallel condition \(AF\parallel OV\), which determines the orbit point on \(t_A\). Fourth, identify the conic through an invariant second focus (ellipse/hyperbola) or invariant directrix (parabola), together with the reflection and gardener/directrix characterization.

This auxiliary-circle framing, imported from the inverse-problem side, gives a concrete pointwise forward construction: each proxy point \(V\) determines a unique orbital point \(A\), rather than relying on a tangent-envelope argument. In this sense, the geometric map from hodograph data to orbit position is explicit at every step.

The same framework also connects cleanly to physical parameters. Equation~(7.1) comes directly from Section~6: \(u=\mu/L\) is the hodograph-circle radius from \Cref{eq:hodo-step3}, and \(p=L^2/\mu\) is the conic polar coefficient from \Cref{eq:proof1f-r-alpha}. Keep only the invariants \(\mu\) and \(L\), then
\begin{equation}
  u=\frac{\mu}{L},
  \qquad
  p=\frac{L^2}{\mu},
\end{equation}
with \(p=\frac{1}{2\kappa}\) from \Cref{eq:kappa-def}.
For ellipse, the equivalent geometric form \(p=b^2/a\) was stated in \Cref{eq:p-b2-over-a-ellipse}.

At an initial orbit point \(A_0\), let
\begin{equation}
  r_0:=FA_0,\qquad d_0:=FH_0,\qquad L=v_0d_0.
\label{eq:proof2f-disc-r0d0L}
\end{equation}
Then
\begin{equation}
  v_{r0}^2=v_0^2-\left(\frac{L}{r_0}\right)^2
  =L^2\!\left(\frac{1}{d_0^2}-\frac{1}{r_0^2}\right).
\label{eq:proof2f-disc-vr0}
\end{equation}
From \Cref{eq:proof1f-r-alpha}, \(r=p/(1-e\cos\alpha)\), so at \(A_0\),
\begin{equation}
  e\cos\alpha_0=1-\frac{p}{r_0},
  \qquad
  e\sin\alpha_0=\frac{p}{L}v_{r0},
\label{eq:proof2f-disc-e-alpha0}
\end{equation}
hence
\begin{equation}
  e^2
  =\left(1-\frac{p}{r_0}\right)^2+\frac{p^2}{L^2}v_{r0}^2
  =1+\frac{p^2}{d_0^2}-\frac{2p}{r_0}.
\label{eq:proof2f-disc-e2-initial}
\end{equation}
Define
\begin{equation}
  E_0:=\frac{v_0^2}{2}-\frac{\mu}{r_0}
  =\frac{L^2}{2d_0^2}-\frac{\mu}{r_0},
\label{eq:proof2f-disc-E0}
\end{equation}
where \(L=v_0d_0\).
Then
\begin{equation}
  e^2=1+\frac{2E_0L^2}{\mu^2},
\label{eq:proof2f-disc-e2-E0}
\end{equation}
so the regime is
\begin{equation}
  E_0<0\ (\text{ellipse}),\qquad
  E_0=0\ (\text{parabola}),\qquad
  E_0>0\ (\text{hyperbola}),
\end{equation}
consistent with \Cref{eq:proof1f-regimes}.

For ellipse/hyperbola, \(e=c/a\), \(p=b^2/a\), and \(s_{\mathrm{hodo}}=b^2/L\). For parabola (with \(a:=MF\) in \Cref{subsec:parabola-proof2f}), \(p=2a\) and \(s_{\mathrm{hodo}}=a^2/L\). Using \(2E_0=v_0^2-\frac{2\mu}{r_0}\), the initial-data form simplifies to
\begin{equation}
  s_{\mathrm{hodo}}
  =
  \begin{cases}
    -\dfrac{L}{2E_0}, & E_0<0,\\[8pt]
    \dfrac{L^3}{4\mu^2}, & E_0=0,\\[8pt]
    \dfrac{L}{2E_0}, & E_0>0.
  \end{cases}
\label{eq:proof2f-disc-shodo-E0}
\end{equation}
Using \(u=\mu/L\), this gives
\begin{equation}
  a=s_{\mathrm{hodo}}\,u
  =
  \begin{cases}
    -\dfrac{\mu}{2E_0}, & E_0<0,\\[8pt]
    \dfrac{L^2}{4\mu}, & E_0=0,\\[8pt]
    \dfrac{\mu}{2E_0}, & E_0>0.
  \end{cases}
\label{eq:proof2f-disc-shodo-u-E0}
\end{equation}
For non-parabolic conics \((E_0\neq 0)\), this is equivalent to
\begin{equation}
  a=\frac{\mu}{2|E_0|},
  \qquad
  b^2=ap=\frac{L^2}{2|E_0|},
  \qquad
  c=ae,\quad
  e=\sqrt{1+\frac{2E_0L^2}{\mu^2}},
\label{eq:proof2f-disc-abc-from-initial}
\end{equation}
so \(a,b,c\) are all fixed by initial data \((\mu,L,r_0,d_0)\).
Equation~\eqref{eq:proof2f-disc-shodo-u-E0} is equivalently the hodograph-scaling statement emphasized in van Haandel--Heckman and then used directly in the CRS treatment \cite{vanhaandelheckman2009teaching,carinena2016newlook}.
So the geometric scaling in Proof~2F is fixed directly by initial dynamical data.

\subsubsection{Orbit-Wide Form and Conserved Specific Energy}

The derivation above is not tied to the special notation \(r_0,v_0,d_0\). Equations~\eqref{eq:proof2f-disc-vr0}--\eqref{eq:proof2f-disc-e2-E0} apply at any point \(A\) on the same orbit by replacing
\[
  (r_0,v_0,d_0,E_0)\ \mapsto\ (r,v,d,E),
\]
with
\[
  L=vd,\qquad E:=\frac{v^2}{2}-\frac{\mu}{r}.
\]
Then the same algebra gives
\begin{equation}
  e^2=1+\frac{2EL^2}{\mu^2}.
\label{eq:proof2f-disc-e2-E}
\end{equation}
Since \(\mu\) is fixed by the force law, \(L\) is fixed by central-force areal invariance, and \(e\) is the global eccentricity of one conic orbit, \Cref{eq:proof2f-disc-e2-E} implies that \(E\) is constant along the orbit. Thus this geometric framework also recovers conservation of specific mechanical energy. Multiplying by the planet mass gives the usual total-energy conservation statement.

This does not claim energy conservation must be derived from \Cref{eq:proof1f-regimes} or from one specific geometric route; conservation of energy is broader. But within the forward-problem discussion, it is useful to state explicitly that the same geometric invariants imply it. Many forward-problem expositions take energy as a starting law (or derive it by integration first, then use geometry), as in the pedagogical lines discussed by Markowsky, van Haandel--Heckman, and Simha \cite{markowsky2011newton,vanhaandelheckman2009teaching,simha2021algebra}.

\subsubsection{Kepler's Third Law and Universal Gravitation}

For the elliptic case, Section~3 gives the area by affine-circle mapping:
\[
  \mathrm{Area}(\text{orbit})=\pi ab.
\]
Since \(L\) is twice areal speed, one full period satisfies
\begin{equation}
  T=\frac{\pi ab}{L/2}=\frac{2\pi ab}{L}.
\label{eq:proof2f-disc-period}
\end{equation}
Using \(b^2=ap\) and \(p=L^2/\mu\),
\begin{equation}
  T^2
  =\frac{4\pi^2a^2b^2}{L^2}
  =\frac{4\pi^2a^3p}{L^2}
  =\frac{4\pi^2}{\mu}\,a^3,
\label{eq:proof2f-disc-third-law}
\end{equation}
hence
\begin{equation}
  \frac{T^2}{a^3}=\frac{4\pi^2}{\mu}=\text{constant}
\label{eq:proof2f-disc-third-law-ratio}
\end{equation}
for all bodies orbiting the same center (same \(\mu\)).

This is the Kepler-third-law form in the present notation. Combined with \( \mathbf F=m\,\mathbf a\), it leads to the inverse-square force model
\[
  \mathbf a=-\frac{\mu}{r^2}\,\hat{\mathbf r},
  \qquad
  \mathbf F=-\frac{\mu\,m_{\mathrm p}}{r^2}\,\hat{\mathbf r}.
\]
By two-body symmetry, \(\mu\) is proportional to the source mass; for solar orbits, \(\mu=G M_\odot\), giving
\[
  \mathbf F=-\frac{G M_\odot m_{\mathrm p}}{r^2}\,\hat{\mathbf r}.
\]
The mass-independence of \(\mathbf a\) in this form is precisely the empirical content of the proportionality between inertial and gravitational mass. Newton's comparison of lunar orbital acceleration with near-Earth free fall is the classical argument for universality. For historical and modern accounts, see the \emph{Principia} sources and reconstructions \cite{newton1687principia,newton1846mottewikisource,newton1934mottecajori,chandrasekhar1995principia,nauenberg2003keplerarea,markowsky2011newton}.

Historically, the route from the third law to gravitation also required a nontrivial geometric step: to place terrestrial gravity and lunar motion under one law, Newton had to control the effects of extended bodies. In the \emph{Principia} he proved the spherical-symmetry (shell) result---an external body is attracted by a spherically symmetric mass as if all its mass were concentrated at the center---thereby making the inverse-square model physically applicable to planets and stars. Coupled with the Moon-versus-surface-fall comparison, this supports the universality claim. The same work also gives geometric treatments of barycentric motion (e.g., Sun--Jupiter), curvature--force relations, and time-of-flight constructions, illustrating how Newton extended Euclidean geometry into dynamics by introducing time into spatial constructions; modern analytical mechanics rewrites much of this in calculus form, but the underlying geometric architecture remains rich and continues to generate interest for deeper dives \cite{newton1687principia,newton1846mottewikisource,newton1934mottecajori,chandrasekhar1995principia,markowsky2011newton}.

To close the forward-problem arc: once the initial data are given, \Cref{eq:proof2f-disc-r0d0L,eq:proof2f-disc-shodo-E0,eq:proof2f-disc-shodo-u-E0,eq:proof2f-disc-abc-from-initial} uniquely fix the rotated/scaled hodograph proxy (its size and orientation) and the corresponding force-center placement in that proxy geometry. From there, the remaining construction follows exactly the same flow as \Cref{subsec:ellipse-proof2f,subsec:hyperbola-proof2f,subsec:parabola-proof2f}. This is why we chose the present order: first show geometrically why conic structure is inevitable, then use initial conditions only to select the exact member (shape and scale) within that conic family.

\section{Summary and Future Work}

This paper develops a Euclidean, \emph{Principia}-style treatment of the Kepler problem using finite-step constructions, tangent/triangle geometry, and Newton's ultimate-ratio viewpoint. For the inverse problem, the proof chain across Sections~2--5 is explicit and modular. Section~2 establishes the dynamical bridge: from constant areal speed, the acceleration direction is central, and two local lemmas convert a geometric ratio into the inverse-square force form. Section~3 introduces the auxiliary-circle/affine transport layer; the tangent-transfer and matching-drop construction reduces local deflection to a conic-invariant ratio, culminating in Inverse Problem Proof~1 for the ellipse (\Cref{subsec:proof1-ellipse-affine}). A more universal proof architecture then generalizes to all conic types: conic-specific geometry computes the local ratio, while the Section~2 conversion step remains universal.

Within that inverse-problem architecture, Proof~2 and Proof~3 play different roles. Proof~2 is based on direct displacement computation, identifying the constant sagitta-to-chord-square ratio. Proof~3 uses a more direct hodograph route: first establish that conic motion gives a circular hodograph, then use \(\Delta \mathbf v/\Delta t\) to recover acceleration, which is inverse-square in radius.

For the forward problem, we label the two forward routes as Proof~1F and Proof~2F, with variant closures denoted by Proof~2F'. In Proof~1F we keep the hodograph-circle theorem in view but proceed in a discrete Newtonian style. The infinitesimal \(\Delta t\)-scaled hodograph circle is translated in parallel as it glides along the orbit: its shape and orientation are preserved while its center shifts from step to step. This moving-circle picture provides a direct geometric bridge from velocity-space circularity to conic recovery in configuration space. 

Our second forward route, Proof~2F, is fully geometric. From the initial condition and fixed physical constants, the hodograph is determined; after rotation and scaling, the remaining steps are ruler-and-compass constructions. The method is pointwise and unique: from local data (the velocity \(\mathbf v\), with both magnitude and direction, together with invariant \(L\)), we set the arm length \(d:=FH=L/v\), identify \(H\), draw \(t_A\perp FH\), and impose the centripetal-direction condition \(AF\parallel OV\). These constraints determine the orbital point \(A\) uniquely in each local configuration. Repeating the construction yields the full orbit. A key invariant is a fixed second focus \(F'\) for ellipse/hyperbola, or a fixed directrix line \(\mathcal D\) for parabola; this certifies the conic class. Variants denoted by Proof~2F' are mainly technical Euclidean line-construction variants that demonstrate flexibility of the method, but are not required for the main architecture.

We also include multiple geometric realizations of the hodograph (and rotated/scaled proxies), including different choices of velocity origin. Several parabola constructions appear to be less emphasized in the existing literature. The force-center auxiliary-circle viewpoint used throughout is also less common than directrix-circle-centered expositions. The infinitesimal \(\Delta t\)-scaled moving-hodograph-circle picture appears to be a useful and potentially original pedagogical viewpoint. More broadly, we intentionally provide a denser geometric presentation than the minimalist style of the \emph{Principia}, supported by modern construction software and figures, to make relation chains explicit and accessible.

Finally, the main limitation remains geometric uniformity across conic types: Inverse Problem Proof~1 (the ellipse affine-circle route) is especially clean, but we have not yet found an equivalent proof of the same type for hyperbola and parabola, which currently require additional case-specific constructions.

\subsection*{Acknowledgment}
This study was motivated by appreciation of geometric beauty and by sustained work on geometric approaches and their extension to broader generalizations, especially through ellipse properties, auxiliary-circle structures, and their related identities. A major precursor was the author's work on an elementary geometry problem for middle-school students, documented in \cite{cryptodogares2026geometrydegreefreedom}.
That direction led to a deeper reading of Newton's \emph{Principia} treatment of space-time Euclidean geometry, and then to possible new exploration paths for both inverse and forward problems.

The author is deeply grateful to the long geometric tradition represented by Euclid, Galileo, Kepler, Apollonius, Newton, Leibniz, Cauchy, Hamilton, Maxwell, Feynman, and Bernoulli, and to modern contributors including Goodstein \& Goodstein, Chandrasekhar, Markowsky, Derbes, van Haandel--Heckman, and Cari\~nena--Ra\~nada--Santander. The manuscript is intended as a pedagogical account for geometric tradition and its application to mechanics.

\subsection*{Code and Tooling Disclosure}
Code, \LaTeX\ sources, and geometric construction files are available at:
\url{https://github.com/CryptoDogAres/AlternativeKeplerToNewton/}

For long-term citation stability, the best practice is to include both the repository URL and the arXiv identifier.

Figures were produced with GeoGebra Classic 6. Writing and editing assistance used OpenAI ChatGPT, Codex, and Prism (including skill-assisted proofreading and \LaTeX\ editing support). Iterative drafting and paper editing were performed in Visual Studio Code.

\bibliographystyle{alphaurl}
\bibliography{bib/references}

\end{document}